\documentclass[reqno]{amsart}
\usepackage{comment}
\usepackage{MnSymbol}  
\usepackage{amsmath}
\usepackage{amsfonts}
\usepackage{enumitem}
\usepackage{array}            
\usepackage{graphicx}
\usepackage{color}
\usepackage{xcolor}
\usepackage{hyperref}
\usepackage{tikz}
\usetikzlibrary{calc}
\usetikzlibrary{positioning}
\usetikzlibrary{patterns}
\usepackage{caption}
\usepackage{subcaption}
\usepackage[justification=centering]{caption}
\usepackage[labelformat=simple,textfont=normalfont,singlelinecheck=off]{subcaption}

\usepackage{marginnote}
\definecolor{mycol}{rgb}{0,0.5,0}

\numberwithin{equation}{section}
\numberwithin{figure}{section}
\numberwithin{subfigure}{figure}
\newcommand{\beal}{\begin{align}}
\newcommand{\eal}{\end{align}}
\newcommand{\beast}{\begin{align*}}
\newcommand{\east}{\end{align*}}
\newcommand{\beqn}{\begin{equation}}
\newcommand{\eeqn}{\end{equation}}
\newcommand{\bee}{\begin{eqnarray}}
\newcommand{\ee}{\end{eqnarray}}
\newcommand{\best}{\begin{eqnarray*}}
\newcommand{\eest}{\end{eqnarray*}}
\newcommand{\bit}{\begin{itemize}}
\newcommand{\eit}{\end{itemize}}
\newcommand{\bem}{\begin{enumerate}}
\newcommand{\eem}{\end{enumerate}}
\newcommand{\bp}{\begin{pmatrix}}
\newcommand{\ep}{\end{pmatrix}}

\newcommand{\lra}{\leftrightarrow}
\newcommand{\mb}{\mathbb}
\newcommand{\ot}[1]{\tilde{#1}}
\newcommand{\oc}[1]{{#1}^{\vee}}
 
\newcommand\Pa{Painlev\'e }
\newcommand\lan{\langle}
\newcommand\ran{\rangle}

\newcommand{\al}{\alpha}

\newcommand{\de}{\delta}
\newcommand{\ta}{\theta}
\newcommand{\De}{\Delta}
\newcommand{\be}{\beta}

\newcommand{\ga}{\gamma}
\newcommand{\Ga}{\Gamma}

\newcommand{\e}{\epsilon}
\newcommand{\oll}[1]{{#1}_{11}}
\newcommand{\ohl}[1]{{#1}_{12}}
\newcommand{\ohll}[1]{{#1}_{112}}
\newcommand{\ol}[1]{{#1}_1}

\newcommand{\hs}{\hspace}

\theoremstyle{break}    \newtheorem{Cor}{Corollary}
\tikzstyle{mybox} = [draw=black,fill=white, very thick,
    rectangle, rounded corners, inner sep=2.5pt, inner ysep=1.5pt]
\tikzset{
root/.style={circle,draw=red!70,fill=red!30, minimum size=14pt,inner sep=0.5pt},
leaf/.style={circle,draw=blue!70,fill=blue!30,minimum size=14pt,inner sep=0.5pt},
gr/.style={circle,draw=green!50!blue!,fill=green!60!blue!30!gray,minimum size=14pt,inner sep=0.5pt},
label/.style={text=blue, font=\bfseries}
}
\tikzset{
tedge/.style={
    double distance=1.0mm+\pgflinewidth,
    postaction={draw}, 
  }
  }

\colorlet{notgreen}{gray!70!orange}
\colorlet{notred}{red!60}
\newtheorem{Rem}{Remark} \theoremstyle{marginbreak}
\newtheorem{lemma}[Cor]{Lemma}
\newtheorem{Def}[Cor]{Definition}
\newtheorem{prop}{Proposition}
\newtheorem{theorem}{Theorem}
\newtheorem{eg}{Example}

\newcommand{\eqn}[1]{(\ref{#1})}
\numberwithin{equation}{section}
\numberwithin{figure}{section}
\numberwithin{prop}{section}
\numberwithin{Rem}{section}
\numberwithin{Cor}{section}
\numberwithin{theorem}{section}
\numberwithin{eg}{section}

\begin{document}

\title{Reflection groups and discrete integrable systems}

\author{Nalini Joshi}\thanks{This research was supported by an Australian Laureate Fellowship \# FL 120100094 and grant \# DP130100967 from the Australian Research Council. }
\address{School of Mathematics and Statistics F07, The University of Sydney, NSW 2006, Australia\\ Tel: +61 2 9351 2172\\ Fax: +61 2 9351 4534}
\email{nalini.joshi@sydney.edu.au}

\author{Nobutaka Nakazono}
\address{School of Mathematics and Statistics F07, The University of Sydney, NSW 2006, Australia\\ Fax: +61 2 9351 4534}
\email{nobua.n1222@gmail.com}

\author{Yang Shi}
\address{School of Mathematics and Statistics F07, The University of Sydney, NSW 2006, Australia\\ Fax: +61 2 9351 4534}
\email{yshi7200@gmail.com}

\maketitle
\date{}    
\begin{abstract}
We present a method of constructing discrete integrable systems with 
crystallographic reflection group
(Weyl) symmetries, thus clarifying the relationship between 
different discrete integrable systems in terms of their symmetry groups.
Discrete integrable systems are associated with space-filling polytopes arise from the geometric representation of the Weyl groups in 
the $n$-dimensional real Euclidean space $\mb{R}^n$. 
The ``multi-dimensional consistency'' property of the discrete integrable system is shown to be inherited from the combinatorial properties of the polytope; while the dynamics of the system is described by the affine translations
of the polytopes on the weight lattices of the Weyl groups.
The connections between some well-known
discrete systems such as the multi-dimensional consistent systems of quad-equations \cite{abs:03} and 
discrete \Pa equations \cite{sak:01} are obtained via the geometric constraints that
relate the polytope of one symmetry group to that of another symmetry group,
a procedure which we call geometric reduction.
\end{abstract}
\tableofcontents
\addtocontents{toc}{\protect\setcounter{tocdepth}{1}}
\section{Introduction}\label{intro}
To find reductions of discrete dynamical systems, regardless of dimension, is an aim which has stimulated a great deal of research. 
We present a constructive geometric method to answer this question based on associating each system with polytopes in $n$-dimensional lattices. 

Focusing in particular on discrete integrable systems, where geometric representations of the Weyl groups are fundamental, the connections between different systems are identified with relations between different Weyl groups, in particular in terms of their reflection subgroups.   In this way, we find geometric reductions from partial difference equations posed on an $n$-dimensional quadrilateral lattice (known as quad-equations)  to Sakai's formulation of the discrete \Pa equations \cite{sak:01}, which are ordinary difference equations.  

While there have been several examples 
associating discrete integrable systems with regular polytopes of the Weyl groups in the literature, they have 
either not been space-filling or have
been restricted to hypercubes. For example, Hirota's d-KP equation is a six-point equation associated with the vertices of an octahedron 
\cite{hirota:81, miwa:82, schief:03, Doliwa:11, akt:15}; the ABS classification of four-point partial difference equations is associated with quadrilaterals
\cite{abs:03}, while six-point equations are associated with the octahedron (octahedron-equations)\cite{abs:12}; and the quadrilateral Yang-Baxter maps have
variables  associated with the edges of a $3$-cube \cite{PTV:06, akt:14}. In contrast, our approach extends to infinite-dimensional affine Weyl groups, which give rise to space-filling polytopes, given by translations of their Voronoi cells.
The vertices of the Voronoi cell are some set of weights of the Weyl group.

This consideration has the advantage that the dynamics of the associated discrete system can then be easily 
understood as a tessellation of Euclidean space by translations of the Voronoi cell on the weight lattice. 
Polytopes associated with the roots of the Weyl groups,
called Delaunay cells, are in general not space-filling polytopes.
Another reason for considering the weights of the Weyl groups is that
the fundamental object in the theory of discrete \Pa equations, 
namely the $\tau$ functions are associated with the weights of the Weyl group \cite{NY:98, ORG:01, KMNOY:06}.

Systems of quad-equations in the literature have been mainly 
defined on the vertices of the $n$-dimensional hypercube ($n$-cube) \cite{abs:03}, or constructed as a consistent system of the same equations \cite{akt:15}. 
In contrast, the systems we consider have a wide variety of symmetries, related via reductions 
to the types of the discrete \Pa equations in Sakai's classification: 19 types of the discrete \Pa equations, 
which follows a degeneration pattern of the affine Weyl symmetry groups \cite{sak:01}: from the root system of type $E_8^{(1)}$
down to $A_1^{(1)}$ (see Figure \ref{N2}). 
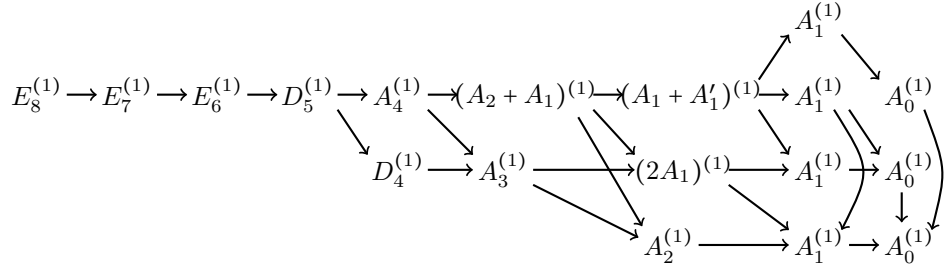
\begin{figure}[t]
\begin{center}
\begin{tikzpicture}[scale = 1]
\begin{scope}
\coordinate (P11s) at (0,0);
\coordinate (P11e) at ($(P11s)+(0.4,0)$);
\coordinate (P12s) at ($(P11e)+(0.8,0)$);
\coordinate (P12e) at ($(P12s)+(0.4,0)$);
\coordinate (P13s) at ($(P12e)+(0.8,0)$);
\coordinate (P13e) at ($(P13s)+(0.4,0)$);
\coordinate (P14s) at ($(P13e)+(0.8,0)$);
\coordinate (P14e) at ($(P14s)+(0.4,0)$);
\coordinate (P15s) at ($(P14e)+(0.8,0)$);
\coordinate (P15e) at ($(P15s)+(0.4,0)$);
\coordinate (P16s) at ($(P15e)+(1.8,0)$);
\coordinate (P16e) at ($(P16s)+(0.4,0)$);
\coordinate (P17s) at ($(P16e)+(1.8,0)$);
\coordinate (P17e) at ($(P17s)+(0.4,0)$);
\coordinate (P18s) at ($(P17e)+(0.8,0)$);
\coordinate (P18e) at ($(P18s)+(0.4,0)$);
\coordinate (P19s) at ($(P18e)+(0.8,0)$);
\coordinate (P19e) at ($(P19s)+(0.4,0)$);
\coordinate (P21s) at (0,-1);
\coordinate (P21e) at ($(P21s)+(0.4,0)$);
\coordinate (P22s) at ($(P21e)+(0.8,0)$);
\coordinate (P22e) at ($(P22s)+(0.4,0)$);
\coordinate (P23s) at ($(P22e)+(0.8,0)$);
\coordinate (P23e) at ($(P23s)+(0.4,0)$);
\coordinate (P24s) at ($(P23e)+(0.8,0)$);
\coordinate (P24e) at ($(P24s)+(0.4,0)$);
\coordinate (P25s) at ($(P24e)+(0.8,0)$);
\coordinate (P25e) at ($(P25s)+(0.4,0)$);
\coordinate (P26s) at ($(P25e)+(1.8,0)$);
\coordinate (P26e) at ($(P26s)+(0.4,0)$);
\coordinate (P27s) at ($(P26e)+(1.8,0)$);
\coordinate (P27e) at ($(P27s)+(0.4,0)$);
\coordinate (P28s) at ($(P27e)+(0.8,0)$);
\coordinate (P28e) at ($(P28s)+(0.4,0)$);
\coordinate (P29s) at ($(P28e)+(0.8,0)$);
\coordinate (P29e) at ($(P29s)+(0.4,0)$);
\coordinate (P31s) at (0,-2);
\coordinate (P31e) at ($(P31s)+(0.4,0)$);
\coordinate (P32s) at ($(P31e)+(0.8,0)$);
\coordinate (P32e) at ($(P32s)+(0.4,0)$);
\coordinate (P33s) at ($(P32e)+(0.8,0)$);
\coordinate (P33e) at ($(P33s)+(0.4,0)$);
\coordinate (P34s) at ($(P33e)+(0.8,0)$);
\coordinate (P34e) at ($(P34s)+(0.4,0)$);
\coordinate (P35s) at ($(P34e)+(0.8,0)$);
\coordinate (P35e) at ($(P35s)+(0.4,0)$);
\coordinate (P36s) at ($(P35e)+(1.0,0)$);
\coordinate (P36e) at ($(P36s)+(1.2,0)$);
\coordinate (P37s) at ($(P36e)+(1.4,0)$);
\coordinate (P37e) at ($(P37s)+(0.8,0)$);
\coordinate (P38s) at ($(P37e)+(0.8,0)$);
\coordinate (P38e) at ($(P38s)+(0.4,0)$);
\coordinate (P39s) at ($(P38e)+(0.8,0)$);
\coordinate (P39e) at ($(P39s)+(0.4,0)$);
\coordinate (P41s) at (0,-3);
\coordinate (P41e) at ($(P41s)+(0.4,0)$);
\coordinate (P42s) at ($(P41e)+(0.8,0)$);
\coordinate (P42e) at ($(P42s)+(0.4,0)$);
\coordinate (P43s) at ($(P42e)+(0.8,0)$);
\coordinate (P43e) at ($(P43s)+(0.4,0)$);
\coordinate (P44s) at ($(P43e)+(0.8,0)$);
\coordinate (P44e) at ($(P44s)+(0.4,0)$);
\coordinate (P45s) at ($(P44e)+(0.8,0)$);
\coordinate (P45e) at ($(P45s)+(0.4,0)$);
\coordinate (P46s) at ($(P45e)+(1.8,0)$);
\coordinate (P46e) at ($(P46s)+(0.4,0)$);
\coordinate (P47s) at ($(P46e)+(1.0,0)$);
\coordinate (P47e) at ($(P47s)+(1.2,0)$);
\coordinate (P48s) at ($(P47e)+(0.8,0)$);
\coordinate (P48e) at ($(P48s)+(0.4,0)$);
\coordinate (P49s) at ($(P48e)+(0.8,0)$);
\coordinate (P49e) at ($(P49s)+(0.4,0)$);
\node at ($(P18s)-(0.4,0)$){$A_1^{(1)}$};
\node at ($(P21s)-(0.4,0)$){$E_8^{(1)}$};
\node at ($(P22s)-(0.4,0)$){$E_7^{(1)}$};
\node at ($(P23s)-(0.4,0)$){$E_6^{(1)}$};
\node at ($(P24s)-(0.4,0)$){$D_5^{(1)}$};
\node at ($(P25s)-(0.4,0)$){$A_4^{(1)}$};
\node at ($(P26s)-(0.9,0)$){$(A_2+A_1)^{(1)}$};
\node at ($(P27s)-(0.9,0)$){$(A_1+A'_1)^{(1)}$};
\node at ($(P28s)-(0.4,0)$){$A_1^{(1)}$};
\node at ($(P29s)-(0.4,0)$){$A_0^{(1)}$};
\node at ($(P35s)-(0.4,0)$){$D_4^{(1)}$};
\node at ($(P36s)-(0.4,0)$){$A_3^{(1)}$};
\node at ($(P37s)-(0.6,0)$){$(2A_1)^{(1)}$};
\node at ($(P38s)-(0.4,0)$){$A_1^{(1)}$};
\node at ($(P39s)-(0.4,0)$){$A_0^{(1)}$};
\node at ($(P47s)-(0.4,0)$){$A_2^{(1)}$};
\node at ($(P48s)-(0.4,0)$){$A_1^{(1)}$};
\node at ($(P49s)-(0.4,0)$){$A_0^{(1)}$};
\draw [->, thick] (P21s)--(P21e);
\draw [->, thick] (P22s)--(P22e);
\draw [->, thick] (P23s)--(P23e);
\draw [->, thick] (P24s)--(P24e);
\draw [->, thick] (P25s)--(P25e);
\draw [->, thick] (P26s)--(P26e);
\draw [->, thick] (P27s)--(P27e);
\draw [->, thick] (P35s)--($(P35e)+(0.2,0)$);
\draw [->, thick] (P36s)--($(P36e)+(0.15,0)$);
\draw [->, thick] (P37s)--(P37e);
\draw [->, thick] (P38s)--(P38e);
\draw [->, thick] (P47s)--(P47e);
\draw [->, thick] (P48s)--(P48e);
\draw [->, thick] ($(P27s)+(0,0.2)$)--($(P17e)-(0,0.2)$);
\draw [->, thick] ($(P18s)-(0.1,0.2)$)--($(P28e)+(0,0.2)$);
\draw [->, thick] ($(P24s)-(0,0.2)$)--($(P34e)+(0,0.2)$);
\draw [->, thick] ($(P25s)-(0,0.2)$)--($(P35e)+(0.2,0.2)$);
\draw [->, thick] ($(P26s)-(0,0.2)$)--($(P36e)+(0.15,0.2)$);
\draw [->, thick] ($(P27s)-(0,0.2)$)--($(P37e)+(0,0.2)$);
\draw [->, thick] ($(P28s)-(0,0.2)$)--($(P38e)+(0,0.2)$);
\draw [->, thick] ($(P36s)-(0,0.2)$)--($(P46e)+(0.2,0.1)$);
\draw [->, thick] ($(P37s)-(0,0.2)$)--($(P47e)+(0,0.2)$);
\draw [->, thick] ($(P38e)+(0.3,-0.3)$)--($(P48e)+(0.3,0.3)$);
\draw [->, thick] ($(P26s)-(0.2,0.3)$)--($(P46e)+(0.25,0.25)$);
\draw[->, thick] ($(P28s)-(0.2,0.2)$) .. controls ($(P38s)+(0.3,0)$) .. ($(P48s)+(-0.1,0.2)$);
\draw[->, thick] ($(P29s)-(0.2,0.2)$) .. controls ($(P39s)+(0.1,0)$) .. ($(P49s)+(-0.1,0.2)$);

\end{scope}
\end{tikzpicture}
\caption{Degeneration of type of symmetries. Of the 22 types listed, 19 (three kinds of discrete equations for the $E_8$ type,
two kinds each for the types $E_7$ and $E_6$, and excluding the
three affine root systems of type $A_0$ on the last column) correspond to the symmetries of
discrete \Pa equations \cite{sak:01}. }\label{N2}
\end{center}
\end{figure}

This means that we have on our hands a corresponding class of quad-equations with immensely  
rich combinatorial/geometrical structures. This paper reviews some of our earlier works in this direction and explain the method in detail with two illustrating
examples.

The plan of the paper is as follows.
In Section \ref{Vor}, we give some preliminaries on the Weyl groups. 
In particular, the actions of the translational elements of the affine Weyl group is discussed in detail.
We define the Voronoi cell of types $A$ and $B$ of the Weyl groups,
which will be central to our construction of the discrete integrable systems in this exposition. In particular, we explain a relation
between the Voronoi cells of types $A$ and $B$ through projection. In Section \ref{DIS}, we construct systems of discrete equations on the Voronoi cells of types $A$ and $B$; and obtain explicit relations between two systems by exploiting the 
geometric relation obtained in Section \ref{Vor}. Two examples of different combinatorial constructions are presented: (i) a system of quad-equations with $W(B_3)$ symmetry and its reduction to 
 a $q$-discrete \Pa type equation with $\widetilde{W}\bigl(A_2^{(1)}\bigr)$ symmetry; (ii) a system of quad-equations 
 with $W(B_2+A_1)$ symmetry and its reduction to 
 a $q$-discrete \Pa type equation with $\widetilde{W}\bigl((A_1+A'_1)^{(1)}\bigr)$ symmetry.
In Section \ref{ag}, 
we construct the rational surface associated with the $(A_2+A_1)^{(1)}$-type root system,
which define the space of initial values of a discrete \Pa equation. 
We discuss the singularity structures of the discrete \Pa equations via intersection theories of 
rational surfaces and the affine Weyl groups on the Picard lattice. 
Finally, in Section \ref{Con} we give some concluding remarks, comment on some of the implications of our result 
and some future directions.

\section{Reflection groups and their associated polytopes}\label{Vor}
We give the necessary properties and facts of the Weyl groups, assuming the
reader is familiar with the theory of the irreducible root systems and their affine
extensions. We follow closely the terminology and notation of Humphreys \cite{Hbook}.
In particular, we discuss the two pictures
associated to the actions of the translational elements of the affine Weyl group via a geometrical representation using a 
$(n+1)$-dimensional real vector space $V^{(1)}$ and its dual space $V^{(1)\ast}$.  The two pictures manifest as 
two complementing aspects of our construction of the discrete integrable systems as will be seen in the later sections
of the paper.
\subsection{Irreducible finite root systems and the corresponding Weyl groups}\label{fW}
Let
$A=(A_{ij})_{1\leq i, j \leq n}$ be a Cartan matrix of type 
$A_n$, $B_n$, $C_n$, $D_n$, $F_4$, $G_2$, $E_6$, $E_7$ or $E_8$ 
\cite[Appendices I-X]{BH}.
Let $V$ and $V^\ast$ be $n$-dimensional real vector spaces spanned by
\beqn
\De=\{\al_1, ..., \al_n\} \quad \mbox{and} \quad \oc \De=\{\oc\al_1, ..., \oc\al_n\},
\eeqn
respectively, and define a bilinear pairing
$\langle {}\,, {} \rangle:$ $V\times V^{\ast} \to \mb R$ by
\beqn\label{Cartan}
\lan\al_i,\oc \al_j\ran=A_{ij}
\eeqn
for all $i, j\in \{1, ..., n\}$.
Since $A$ is non singular, $V^\ast$ is isomorphic to the dual space of $V$.
The elements of $\De$ and $\oc \De$
are called the simple roots and simple coroots; and
\beqn
Q=\mb Z \De \quad \mbox{and} \quad \oc Q=\mb Z \oc \De,
\eeqn
are the root lattice and coroot lattice. The height of the vector $\sum \lambda_i \al_i\in Q$
is defined to be $\sum \lambda_i$.

For each $i\in \{1, ..., n\}$ define a linear transformation $s_{\al_i}: V\to V$ by the requirement that
\beqn\label{siC}
s_{\al_i}\al_j=\al_j-\lan\al_j,\oc \al_i\ran\al_i=\al_j-A_{ji}\al_i.
\eeqn
It is well known that the linear group $W'$ generated by $s_{\al_1}$, ..., $s_{\al_n}$
is finite. Furthermore, $W'$ is isomorphic to the Weyl group of the Cartan matrix $A$, defined to be
the abstract group $W=W(A)$ generated
by $s_1$, ..., $s_n$ subject to the defining relations $(s_is_j)^{m_{ij}}=1$, for all $ i, j\in \{1, ..., n\}$,
where $(m_{ij})_{1\leq i, j \leq n}$
is the Coxeter matrix associated with the Cartan matrix $A$ \cite[no. 1.5, Chap. VI] {BH}. Note that
\beqn\label{aijm}
A_{ij}A_{ji}=4 \mbox{ cos}^2 \left(\frac{\pi}{m_{ij}}\right)
\eeqn
for $i\neq j$
and $A_{ii}=2$. (In particular, $m_{ii}=1$ for all $i$, and $m_{ij}\in\{2, 3, 4, 6\}$ for $i\neq j$).
Let $\eta: W \to W'$ be the isomorphism satisfying $\eta s_i=s_{\al_i}$ for all $i\in \{1, ..., n\}$, and
define an action of $W$ on $V$ via $w.v=(\eta w)v$ for all $w\in W$ and $v\in V$.

Since $W'$ is finite it preserves a Euclidean inner product $(\, ,\,): V\times V\to \mb{R}$. If $0\neq \al \in V$, then
the reflection in the hyperplane orthogonal to $\al$ is the orthogonal transformation $V\to V$ given by
\[
v\mapsto v-\frac{2(\al, v)}{(\al, \al)}\al,
\]
for all $v\in V$. Since 
$
(\al_i, \al_j)=(s_{\al_i} \al_i, s_{\al_i} \al_j)=(-\al_i, \al_j-A_{ji}\al_i),
$
it follows that 
\beqn\label{Aaij}
A_{ji}=2\frac{(\al_i,\al_j)}{(\al_i, \al_i)}
\eeqn
for all $i, j\in \{1, ..., n\}$ and that $s_{\al_i}$ is the reflection 
in the hyperplane $H_{i}$:
\beqn\label{HV}
H_{_i}=\{v\in V\mid \lan v, \oc\al_i\ran=0\}=\{v\in V\mid (v, \al_i)=0\}.
\eeqn
Note that Equation \eqn{Aaij} implies that
\beqn\label{AAaij}
(\al_i,\al_i)A_{ji}=A_{ij}(\al_j, \al_j)
\eeqn
for all $i$ and $j$.
In general, we write $s_\al$ for
the reflection in the hyperplane orthogonal to $\al$. A trivial calculation shows
that if $w: V\to V$ preserves $(\, , \, )$ then 
\beqn
s_{w. \al}=w s_\al w^{-1}
\eeqn
for all nonzero $\al\in V$.

The root system of $W$ is defined to be the subset $\Phi$ of $Q$ given by $\Phi=W'\De$. If $\al\in \Phi$, 
then $\al=w\al_i$ for some $w\in W'$ and $i\in \{1, ..., n\}$,
hence $s_\al=w s_{\al_i} w^{-1}\in W'$. For simplicity of notation, we henceforth use the isomorphism $\eta$
to identify elements of $W'$ with elements of $W$, so that $s_i=s_{\al_i}$ and
$s_\al\in W$ whenever $\al\in \Phi$.

Note that $W$ acts on $V^\ast$ via the contragredient action:
\beqn\label{cong}
\lan w^{-1}. f,  h\ran=\lan f, w.h\ran, \quad f\in V, h\in V^\ast, w\in W.
\eeqn
It follows that
\beqn\label{siCi}
s_{i} .\oc\al_j=\oc \al_j-A_{ij}\oc \al_i,
\eeqn
and the coroot system is defined by
$
\oc \Phi=W.\oc \De.
$

Define a linear isomorphism $\ta: V^\ast \to V $ by
$\ta \oc \al_i=\frac{2 \al_i}{(\al_i, \al_i)}$ for all $i$. Then Equations \eqn{AAaij} and \eqn{siC} combine to give
\begin{align*}
s_i. (\ta \oc \al_j) &=s_i. \left(\frac{2 \al_j}{(\al_j, \al_j)}\right)   \\
 &=\frac{2 \al_j}{(\al_j, \al_j)}-\frac{2 }{(\al_j, \al_j)}A_{ji}\al_i\\
 &=\frac{2 \al_j}{(\al_j, \al_j)}-\frac{2 }{(\al_i, \al_i)}A_{ij}\al_i\\
 &=\ta \oc \al_j- A_{ij}\ta\oc \al_i\\
 &=\ta (\oc \al_j- A_{ij}\oc \al_i)\\
 &=\ta (s_i .\oc \al_j),
\end{align*}
where we have used Equation \eqn{siCi} in the last line. It follows that $\ta$ commutes with the $W$-- actions on $V^\ast$ and $V$.
By the definition of $\ta$ we have $\oc \al_i=\ta^{-1}\left(\frac{2 \al_i}{(\al_i, \al_i)}\right)$, and
we now extend this to the entire root system by define 
\beqn\label{ocal}
\oc \al=\ta^{-1}\left(\frac{2 \al}{(\al, \al)}\right)
\eeqn
for all $\al \in \Phi$. Observe that if $\al=w. \al_j$ for some $w\in W$ and $j\in \{1, 2, ..., n\}$ then
\[
\ta (w.\oc \al_j) =w. \ta(\oc \al_j)  
=2\frac{w. \al_j}{(\al_j, \al_j)}
=2\frac{w. \al_j}{(w.\al_j, w.\al_j)}=\frac{2\al}{(\al, \al)}=\ta (\oc \al),
\]
since $W$ preserves the inner product. Hence $\oc \al=w.\oc \al_j$ whenever $\al=w.\al_j$.
Since we are considering only irreducible root systems, $\Phi$ contains a unique root $\ot\al$,
whose height is maximal, called the {\it highest root}. Define $c_1$, $c_2$, ..., $c_n\in \mb{R}$ by
\beqn\label{hr}
\ot{\al}=\sum_{i=1}^{n}c_{i}\al_{i}.
\eeqn

The values of the $c_i$ for all finite Weyl groups can be found in \cite{Hbook}
for example. Similarly in $\oc\Phi$ we have the highest coroot,
\beqn\label{hcr}
\oc{\ot{\al}_s}=\sum_{i=1}^{n}m_{i}\oc\al_{i},
\eeqn
where $\ot{\al}_s$ is the highest short root in $\Phi$,  $m_i=c_i$ for the simply laced root systems (types $A$, $D$, $E$), and
 $m_i=c_{n-i+1}$ otherwise. Note that for the simply laced root systems $\oc{\ot{\al}_s}=\oc{\ot{\al}}$.

Let $h_1$, ..., $h_n \in V^\ast$ satisfy
\beqn\label{h}
\lan \al_i, h_j\ran=\delta_{ij}, \quad (1\leq i, j\leq n),
\eeqn
so that $\{h_1, ..., h_n\}$ is the basis of $V^\ast$ dual to the basis $\{ \al_1, ..., \al_n\}$
of $V$. The elements $h_1$, ..., $h_n$ are called the fundamental weights, and
$P=\mb Z \{h_1, ..., h_n\}$
is called the weight lattice. 
Observe that from Equations \eqn{Cartan} and \eqn{h} we have
\beqn\label{rw}
\oc \al_i=\sum_{i=1}^{n}A_{ij}h_j.
\eeqn
Moreover, 
\beqn\label{sih}
 s_i.h_j=h_j-\lan\al_i,  h_j\ran\oc \al_i=
 \begin{cases}
 h_j, &\text{for}\quad i\neq j,\\
 h_j-\oc \al_j=h_j-\sum_{k=1}^{n}A_{jk}h_k,&\text{for} \quad i=j.
 \end{cases}
 \eeqn
Since in this paper we will be primarily interested in the Weyl group of types $A$ and $B$, we give their definitions below.

\subsection{Combinatorial description of Weyl group of types $A$ and $B$}\label{Bn}
The Weyl group of type $A_{n-1}$, denoted $W(A_{n-1})$,  is isomorphic to $\mathfrak{S}_{n}$, 
the symmetric group 
on $\{1, ..., n\}$,
with $s_i$ acting as the transposition (i, i+1) for each $i \in \{1, ..., n-1\}$.
The Dynkin diagram $\Ga(A_{n-1})$ (shown in Figure \ref{DAn}) describes the relation
satisfied by $s_1, ..., s_{n-1}$ as explained in \cite[page 31]{Hbook},
namely
\bee\label{geneaW}
&&(s_is_j)^2=1\;\mbox{if}\;\mid i-j\mid >1,\quad\mbox{and}\quad
(s_is_{i+1})^3=1\;\mbox{if}\;\mid i-j\mid =1,
\ee
for $i, j\in\{1, ..., n-1\}$.
\begin{figure}
\begin{tikzpicture}
\node  (a1) {$\circ$};
\node [right=of a1](a2) {$\cdots$} ;
\node [right=of a2](a3) {$\circ$} ;
\node [right=of a3](a4) {$\circ$};
\draw (a1) node [anchor=north] {$s_1$} ;
\draw (a3) node [anchor=north] {$s_{n-2}$} ;
\draw (a4) node [anchor=north] {$s_{n-1}$} ;
\draw[-] (a1) -- node {} (a2);
\draw[-] (a2) -- node {} (a3);
\draw[-] (a3) -- node {} (a4);
\path[use as bounding box] (-1.5,0) rectangle (0,0);
\end{tikzpicture}
\caption{The Dynkin diagram $\Ga(A_{n-1})$} \label{DAn}
\end{figure}
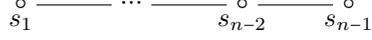

The Weyl group $W(B_{n})=\langle s_1, ..., s_{n-1}, s_n\rangle$ acts on the set of all subsets of $\{1, ..., n\}$,
while the subgroup
$\lan s_1, ..., s_{n-1}\ran$
acting via permutations of  $\{1, ..., n\}$ as in type $A_{n-1}$ above, and with $s_n$ acting via
\beqn
s_n: S\mapsto S\bigominus \{n\}\quad \text{for all $S\subseteq\{1, ..., n\}$},
\eeqn
where $\bigominus$ denotes the symmetric difference, $I\bigominus J=(I\bigcup J)\backslash (I\bigcap J)$.
The Dynkin diagram $\Ga(B_n)$ is shown in Figure \ref{DBn},
and the corresponding relations are:
\bee\label{funWB}
(s_{n-1}s_n)^4=1,\quad
(s_is_j)^2=1\;\mbox{if}\;\mid i-j\mid >1,& (s_is_{i+1})^3=1\;\mbox{if}\;\mid i-j\mid =1,
\ee
for $i, j\in\{1, ..., n-1\}$. Now we are ready to define the Voronoi cell of the Weyl groups.

\begin{figure}
\begin{tikzpicture}
\node  (a1) {$\circ$};
\node [right=of a1](a2) {$\cdots$} ;
\node [right=of a2](a3) {$\circ$} ;
\node [right=of a3](a4) {$\circ$};
\draw (a1) node [anchor=north] {$s_1$} ;
\draw (a3) node [anchor=north] {$s_{n-1}$} ;
\draw (a4) node [anchor=north] {$s_n$} ;
\draw[-] (a1) -- node {} (a2);
\draw[-] (a2) -- node {} (a3);
\draw[double distance=1.5pt] (a3) -- node {} (a4);
\path[use as bounding box] (-1.5,0) rectangle (0,0);
\end{tikzpicture}
\caption{Dynkin diagram $\Ga(B_{n})$} \label{DBn}
\end{figure}

\subsection{Voronoi cell: space-filling Polytopes}
{\Def [\cite{CSbook}]\label{dV}

Let the fundamental simplex $F$ be the convex hull of 
\[\left\{\,{\bf 0},\frac{h_{1}}{m_1}, \frac{h_{2}}{m_2},..., \frac{h_{n}}{m_n}\right\},\]
where $h_i$ are the fundamental weights,
$m_i$ are the coefficients in the expression \eqn{hcr} for the highest coroot.
 The Voronoi cell of the Weyl group $W$, denoted by {\rm Vor}$(W)$ is the union
of the images of the fundamental simplex under $W$,
\beqn
{\rm Vor}(W)=\bigcup_{w\in  W} w .F, \quad w\in W.
\eeqn}

{\theorem [\cite{venkov:54, CSbook}]\label{Vthm}
{\rm Vor}$(W)$ is a space-filling polytope - the whole space ${\mb R}^n$ can be tessellated by translated copies of 
{\rm Vor}$(W)$.
}

The construction of the Voronoi cells for all finite Weyl groups 
 is given in \cite{moody:92} via an approach of the decorated Dynkin diagrams. 
 Here we restrict our attention to types $A$ and $B$.  

{\prop [\cite{moody:92}]\label{actionBn}
 The polytope {\rm Vor}$(B_n)$ is an $n$-cube.}

\begin{proof}
The coefficients $m_i$ for the root system of type $B_n$ are $2, 2, ..., 2, 1$.
The fundamental simplex $F$ is then the convex hull of $\left\{\,0,\frac{h_{1}}{2}, \frac{h_{2}}{2},...,\frac{h_{n-1}}{2}, {h_{n}}\right\}$,
and the polytope {\rm Vor}$(B_n)$ is the convex hull of $Orb(h_n)=W(B_n).h_n$ -- the set of weight vectors in the orbit of the fundamental weight $h_n$ 
under the actions of $W(B_n)$. In the geometric representation of $W(B_n)$
in $\mb{R}^n$, we have
\beqn
h_n=\frac{1}{2}(\e_1+...+\e_n), \eeqn
where $\e_i$ (i=1, ..., n) are
the unit vectors of $\mb{R}^n$.

The actions of the generators of $W(B_n)$ on this geometric representation are that $s_i$ 
permutes $\e_i\leftrightarrow \e_{i+1}$ for $i\in \{1, ..., n-1\}$, and $s_n$ changes the sign of $\e_n$.
Then weights in $Orb(h_n)$ are of the form
\[
\frac{1}{2}
\bp
\pm \e_1+ ... \pm \e_n
\ep.
\]
The $2^n$ weights in this orbit, are exactly the $2^n$ vertices of an $n$-cube with edge length
1, centered at the origin $\bf{0}$.

To simplify the notation, we denote the vectors in $Orb(h_n)$ by $X_S$, $S\subseteq\{1, ..., n\}$. The subset $S$ encodes 
the positive components of the vector. For example, 
$h_n=\frac{1}{2}(\e_1+\e_2+\cdots+\e_n)=X_{\{1, 2, ..., n\}}=X_{12...n}$, and 
$\frac{1}{2}(-\e_1-\e_2-\cdots-\e_n)=X_{\emptyset}=X_0$.  Therefore, the set of vertices of {\rm Vor}$(B_n)$ in this notation is given by
\beqn\label{VorB}
Orb(h_n)=\{X_S\}_{S\subseteq\{1, ..., n\}}.
\eeqn 
\end{proof}

{\prop [\cite{jns2}] \label{bta}  The polytope {\rm Vor}$(A_{n-1})$ is obtained from {\rm Vor}$(B_n)$
by an orthogonal projection $\phi: \mb{R}^n \mapsto \mb{R}^{n-1}$ along the fundamental weight of $W(B_n)$: $h_n=\frac{1}{2}(\e_1+...+\e_n)$,
defined by 
\begin{align}\label{phi}
 \phi \,.v &= v-\frac{(v,h_n)\,h_n}{\| h_n\| ^2}, \\[-2.5em]\nonumber
 \intertext{where $(\, ,\,)$ is the Euclidean inner product and}
 \{U_S\}_{S\subseteq\{1, ..., n\}}&=\{\phi .X_S\},\quad 1\leq |S| \leq n-1,
\end{align}
are the $2^n-2$ vertices of {\rm Vor}$(A_{n-1})$.
}
\begin{proof}
The coefficients $m_i$ for the root system of type $A_{n-1}$ are $\{1, 1, ..., 1, 1\}$.
The fundamental simplex $F$ is then the convex hull of $\left\{\,0,h_1, \, h_2,..., {h_{n-1}}\right\}$,
where
 \beqn\label{hkA}
 h_{k}=(\e_1+\cdots+\e_k)-\frac{k}{n}\sum_{i=1}^{n}\e_{i},\quad 1\leq k\leq n-1,
\eeqn
are the fundamental weights of $W(A_{n-1})$. We denote 
the vectors in $Orb(h_k)$ by $\{U_S\}$ for $S\subseteq\{1, ..., n\}$ and $|S|=k$),
where the subset $S$ in $U_S$ records the positive components of the vector. 
The vertices of {\rm Vor}$(A_{n-1})$,  are the union of the orbits of all the fundamental weights: 
\beqn\label{VorA}
\{U_S\}_{S\subseteq\{1, ..., n\}}=\bigcup_{1\leq k \leq n-1 } r . h_k, \;r\in W(A_{n-1}) .
\eeqn

The projection $\phi$ chooses the $(n-1)$-dimensional hyperplane in $\mb{R}^n$ 
orthogonal to $h_n=\frac{1}{2}(\e_1+...+\e_n)$, which corresponds to the usual 
representation of $A_{n-1}$ in $\mb{R}^n$. It can be checked directly that
$\{\phi .X_S\}=\{U_S\}$. Note that two vertices of {\rm Vor}$(B_{n})$ coincide after the orthogonal projection: 
$\phi .X_{1 2 ... n}=\phi .X_{0}={\bf 0}$, and are now at the center
of {\rm Vor}$(A_{n-1})$.
\end{proof}

{\Rem 
The projection of {\rm Vor}$(B_n)$ to {\rm Vor}$(A_{n-1})$ on the group level corresponds to taking the parabolic subgroup
$W(A_{n-1})$ of $W(B_{n})$. On the Dynkin diagram,  it is associated  with deleting the node $s_n$ from 
$\Ga(B_n)$ (Figure  \ref {DBn}) to obtain $\Ga(A_{n-1})$ (Figure \ref{DAn}).
} 

We give an example of Propositions \ref{actionBn} and \ref{bta} for the case $n=3$.
{\eg \label{B3}
The fundamental weight $h_3$ of $W(B_{3})$ is $\frac{1}{2}(\e_1+\e_2+\e_3)$. The vertices of {\rm Vor}$(B_3)$
are the vectors in $Orb(h_3)$. They are of the form: $\frac{1}{2}(\pm\e_1+\pm\e_2\pm\e_3)$. In the $X_S$ notation they
are
\beqn\label{3cubev}
\{X_{S}\}_{S\subseteq\{1, ..., 3\}} =\{\,X_{0},X_{1},X_{2},X_{3}, X_{12},X_{13},X_{23},X_{123}\}.
\eeqn
They are the $8$ vertices of a $3$-cube.

The generators of $W(B_3)=\langle s_1, s_{2}, s_3\rangle$ act on the variables $X_S$ 
as follows:
\begin{align}\nonumber
&s_1:\{X_1\lra X_2,\quad X_{13}\lra X_{23}\},\\\label{natB3}
&s_2:\{X_2\lra X_3,\quad X_{12}\lra X_{13}\},\\\nonumber
&s_3:\{X_0\lra X_3,\quad X_{1}\lra X_{13},\quad X_{2}\lra X_{23},\quad X_{12}\lra X_{123}\},
\end{align}
satisfying the following relations 
\begin{equation}\label{funWB3}
s_1^2=s_2^2=s_3^2=1,\;\;(s_1s_{2})^3=1, \quad (s_{1}s_3)^2=1,\quad (s_{2}s_3)^4=1,
\end{equation}
which correspond to the Dynkin diagram $\Ga(B_3)$ (see Figure \subref{DB3}).

\begin{figure}
\begin{subfigure}[b]{.4\textwidth}
\centering 
\hs{-1cm}
\begin{tikzpicture}[scale=1]
\node  (a1) {$\circ$};
\node [right=of a1](a2) {$\circ$} ;
\node [right=of a2](a3) {$\circ$};
\draw (a1) node [anchor=north] {$s_1$} ;
\draw (a2) node [anchor=north] {$s_2$} ;
\draw (a3) node [anchor=north] {$s_3$} ;
\draw[-] (a1) -- node {} (a2);
\draw[double] (a2) -- node {} (a3);
\path[use as bounding box] (-1.5,0) rectangle (0,0);
\end{tikzpicture}
\subcaption{$\Ga(B_3)$}
\label{DB3}
\end{subfigure}
\begin{subfigure}[b]{.4\textwidth}
\centering
\hs{-1.5cm}
\begin{tikzpicture}[scale=1]
\node  (a1) {$\circ$};
\node [right=of a1](a2) {} ;
\node [right=of a2](a3) {$\circ$} ;
\draw (a1) node [anchor=north] {$s_1$} ;
\draw (a2) node [anchor=south] {} ;
\draw (a3) node [anchor=north] {$s_{2}$} ;
\draw[-] (a1) -- node {} (a3);
\path[use as bounding box] (-1.5,0) rectangle (0,0);
\end{tikzpicture}
\subcaption{$\Ga(A_2)$}
\label{DA2}
\end{subfigure}
\end{figure}

By Equation \eqn{VorA}, the vertices of {\rm Vor}$(A_2)$ are the weight vectors in the orbits of $h_1$ and $h_2$. 
The fundamental weights of $W(A_2)$ are:
\begin{align}\label{fwA2}
&h_1=U_{1}=\frac{1}{3}(2\e_1-\e_2-\e_3),\\\nonumber
&h_2=U_{12}=\frac{1}{3}(\e_1+\e_2-2\e_3).
\end{align}
The generators of $W(A_2)=\langle s_1, s_{2}\rangle$ act on the variables $U_S$ 
as follows:
\begin{align}\label{natA2}
&s_1:\{U_1\lra U_2,\quad U_{13}\lra U_{23}\},\\\nonumber
&s_2:\{U_2\lra U_3,\quad U_{12}\lra U_{13}\},
\end{align}
satisfying the fundamental relations 
\begin{equation}\label{funWA2}
s_1^2=s_2^2=1,\;\;(s_1s_{2})^3=1, 
\end{equation}
which correspond to the Dynkin diagram $\Ga(A_2)$ (see Figure \subref{DA2}).

{\rm Vor}$(A_2)$ is a hexagon with 6 vertices:
\beqn\label{hexv}
\{U_{S}\}_{S\subseteq\{1, ..., 3\}} =\{\,U_{1}, U_{2}, U_{3}, U_{12}, U_{13}, U_{23}\},
\eeqn
centred at  $U_0=h_0=\bf{0}$. 

It can be checked easily that the set of vertices given in \eqn{hexv} results from the set of vertices given in \eqn{3cubev} by 
orthogonal projection $\phi$ \eqn{phi} along $X_{123}=h_3=\frac{1}{2}(\e_1+\e_2+\e_3)$. Geometrically this corresponds 
to the projection of a $3$-cube to a hexagon.
}

\subsection{Affine Weyl group and translations}\label{twopic}
Given the Cartan matrix $A=(A_{ij})_{1\leq i, j \leq n}$ as in Section \ref{fW} above, we now define

\begin{align}\label{eCno}
 A_{j0}&=\lan \al_j, -\oc{\ot \al}\ran,\quad
 A_{0j}=\lan -\ot \al, \oc \al_j\ran,
\end{align}
and let $A_{00}=2$. The $(n+1)\times (n+1)$ matrix $A^{(1)}=(A_{ij})_{0\leq i, j \leq n}$ is called the extended Cartan matrix.
Let $W\bigl(A^{(1)}\bigr)$ be the abstract group generated by $s_0, s_1, ..., s_{n}$, subject to the defining relations 
$(s_is_j)^{m_{ij}}=1$ for all $i, j \in \{0, 1, ...n\}$, where the $m_{ij}$ are given by Equation \eqn{aijm}. These $m_{ij}$ can
be read off from the {\it extended Dynkin diagrams} $\Ga^{(1)}$ see \cite[p.34]{Hbook}.
The group $W\bigl(A^{(1)}\bigr)$ is called the affine Weyl group of type ~$A$.  

The corresponding extended vector space $V^{(1)}$ is spanned by the set of simple affine roots, $\De^{(1)}=\De\cup\{\al_0\}$.
In addition, we write $s_0$ for the reflection on $V^{(1)}$ corresponding to $\al_0$. Thus 
\beqn\label{sia}
 s_i.\al_j=
\al_j-A_{ji} \al_i \quad \text{for all} \;i, j \in \{0, 1, ...n\}.
 \eeqn
 We define
$
 \Phi^{(1)}=W(A^{(1)}).\De^{(1)},
$
and call $\Phi^{(1)}$ the {\it the affine root system}. 
It can be shown that
\begin{align}
 \Phi^{(1)}&=\{\,\al+k\de \mid  \al\in \Phi, k\in \mb{Z}\},
\end{align}
where 
\beqn\label{de}
 \de=\al_0+\ot \al=\sum_{i=0}^{n}c_{i}\al_{i}
\eeqn
is called ${\it null \;root}$. We have let $c_0=1$, and $c_i$ ($i\in \{1, ..., n\}$) are defined earlier.
From the fact that $A^{(1)}$ is singular we have
\beqn\label{detr}
\lan \de, \oc \al_j\ran=\lan \sum_{i=0}^{n}c_{i}\al_{i}, \oc \al_j\ran= \sum_{i=0}^{n}c_{i}A_{ij}=0, \quad \mbox{for} \; 0\leq j\leq n,
\eeqn
and that all elements of $W$ act trivially on $\de$.
The extended space $V^{(1)}$ is spanned by the simple affine roots $\{  \al_0, \al_1, ...,  \al_n\}$, or
equally well, by the basis $\{ \al_1, ..., \al_n,  \de\}$.
We extend the bilinear pairing in Equation \eqn{h} by
$\langle {}\,, {} \rangle$: $V^{(1)}\times V^{(1)\ast} \to \mb R$ by
\beqn
\lan \al_i, h_j\ran=\delta_{ij}, \quad\mbox{for}\quad (1\leq i, j\leq n),\quad\mbox{and}\quad\lan \de,h_\de \ran=1.
\eeqn
 Then $\{\, h_1, ... ,  h_n,  h_\de\}$ is the basis of  $V^{(1)\ast}$ 
dual to the basis $\{ \al_1, ...,  \al_n, \de\}$ of $V^{(1)}$. 

There are two pictures of the actions of $W\bigl(A^{(1)}\bigr)$. In one they act linearly on the simple affine roots  $\{ \al_0,  \al_1, ..., \al_n\}$ in $V^{(1)}$ by formula \eqn{sia} for the reflections $s_i$ ($i\in \{0,1, ..., n\}$). In  the second picture,
the actions can be seen as affine transformations in a hyperplane of $V^{(1)\ast}$,
which give rise to translational motions that are nonlinear. We give a detailed explanation below.
Define a hyperplane $H$ in $V^{(1)\ast}$ by,
\beqn\label{H}
H=\{\,h\in V^{(1)\ast}\,|\, \lan \de, h\ran=1\}.
\eeqn
Then generators $s_{\al-k\de}$ associated with the affine roots $\al-k\de\in \Phi^{(1)}$ act on $h \in H$ by the formula
 \beqn\label{sakd}
s_{\al-k\de}.h=h-\left( \lan \al, h\ran-k\right) \oc \al,
\eeqn
where $\oc\al\in \oc \Phi$, $k\in \mb{N}$.
The generator $s_{\al-k\de}$ describes reflection about the hyperplane $H_{\al-k\de}$ defined by , 
\beqn\label{Hk}
H_{\al-k\de}=\{\,h \in H\; |\; \lan \al, h\ran=k\}, \quad\mbox{note that}\quad  H_{\al-k\de}=H_{-\al+k\de}.
\eeqn
Observe that $s_{\al-k\de}$ defined in Equation \eqn{sakd} are reflections about reflection planes
that do not pass through the origin for $k\neq 0$.
The hyperplane $H$ is called an {\it affine plane}.
In particular, we see that
\beqn\label{sakd0}
s_0. h=s_{\al_0}.h=s_{\de-\ot\al}.h=s_{\ot\al-\de}.h=h-\left(\lan \ot\al, h\ran-1\right)\oc {\ot\al},
\eeqn
and it follows that
\begin{align*}
 &t_{\ot \al}=s_{0}s_{\ot\al}. h=s_{\de-\ot\al}s_{\ot\al}.h=h+\oc {\ot\al},
\end{align*} 
where we have denoted the element of $W\bigl(A^{(1)}\bigr)$ associated with translation of $\oc {\ot\al}$ on $H$ as $t_{\ot \al}$. 
It is known that the affine Weyl group decomposes into the semidirect product of translations in the coroot lattice
and the finite Weyl group:
\beqn
W\bigl(A^{(1)}\bigr)=\langle t_{\al_1}, ..., t_{\al_n}\rangle\rtimes W(A)=T_{\oc Q} \rtimes W(A),
\eeqn
where $t_{ \al_i}\in W\bigl(A^{(1)}\bigr)$ for $(i=1, ..., n)$, $ \al_i\in  \De$. The group
$W\bigl(A^{(1)}\bigr)$ can be further extended
by Dynkin diagram automorphisms to the $\it{extended}$ affine Weyl group $\widetilde{W}(A^{(1)})$, which  
decomposes into the semidirect product of the translations in the weight lattice
and the finite Weyl group acting on the weight lattice:
\beqn
\widetilde{W}(A^{(1)})=\langle t_{h_1}, ..., t_{h_n}\rangle\rtimes W(A)=T_{P} \rtimes W(A),
\eeqn
where $t_{h_i}\in \widetilde{W}(A^{(1)})$ $(i=1, ..., n)$ are the translational elements 
associated to the fundamental weights. Note that we denote the Dynkin diagram of $\widetilde{W}(A^{(1)})$ as $\widetilde{\Ga}^{(1)}(A)$.

Now we give the definition of the translational elements and their actions on $V^{(1)}$ and $H\subset V^{(1)\ast}$ respectively.
\subsubsection{Translational elements of the affine and extended affine Weyl groups}

{\Def [\cite{kac:90}] Let $\mu\in V$, so that $\lan\de, \oc\mu\ran=0$ and $\lan \mu, \oc \mu\ran\neq0$.
Translational element associated to $\mu$ in $W(A^{(1)})$ is given by

\begin{align}\nonumber%
 t_{\mu}&=s_{\de-\mu}s_{\mu},\\[-1.5em]\label{transs}
 \intertext{ and}\label{tat}
wt_\mu w^{-1}&=t_{w.\mu},\quad w\in W(A^{(1)}).
\end{align}
The action of $t_\mu$ on $H$ is given by
\beqn\label{Transn}
t_{\mu}. h= h+\oc\mu, 
\eeqn
 }
and its action on the simple affine roots $\al_{i} \in V^{(1)}$ ($i\in \{0,1, ..., n\}$) is given by
 \beqn\label{Transl}
 t_{\mu}.\al_i=\al_i-\lan\al_i, \oc\mu\ran\de=\al_i-\mu_i \de,
 \eeqn
where we have let $\lan\al_i, \oc\mu\ran=\mu_i$. Hence $t_{\mu}$ shifts $\al_i$ by $-\mu_i$ multiples of $\de$.
Observe that the $\mu_i$ ($i\in \{0,1, ..., n\}$) satisfy the constraint
\begin{align}\label{Tranc}
  0=\lan\de, \oc\mu\ran&=\lan\sum_{i=0}^{n}c_{i}\al_{i}, \oc\mu\ran=\sum_{i=0}^{n}c_{i}\mu_i.
\end{align}
We refer to the linear action of $t_{\mu}$ on $\al_{i} \in V^{(1)}$ ($i\in \{0,1, ..., n\}$) given by Equation \eqn{Transl} as {\it shift motion} on the simple affine roots.

In what follows we give an explicit example of the ideas explained above
for the root system of type $A_2$.

{\eg \label{TranA2}
Affine and extended affine Weyl group of type $A_2$
are: ${W}(A_2^{(1)})=\langle s_0, s_1, s_2 \rangle$ and $\widetilde{W}(A_2^{(1)})=\langle s_0, s_1, s_2, \rho\rangle$, respectively, where $\rho$ is
the Dynkin diagram automorphism:
\beqn\label{rhor}
\rho : (\al_0,\al_1, \al_2)\mapsto (\al_1,\al_2, \al_0).
\eeqn
The generators satisfy the following relations
\bee\label{funWaA2}
&&s_j^2=1,\;\;(s_js_{j+1})^3=1,\;\;(j\in \mb Z/3\mb Z),\\\nonumber
&&\rho^3=1, \;\;\rho s_j=s_{j+1}\rho,
\ee
which correspond to the Dynkin diagram in Figure \ref{DaA21}.    
     
\begin{figure}     
\begin{tikzpicture}
\node  (a1) {$\circ$};
\node [right=of a1](a2) {} ;
\node [right=of a2](a3) {$\circ$} ;
\node (a4) at (1.4, -0.9){$\circ$} ;
\draw (a1) node [anchor=south] {$s_1$} ;
\draw (a2) node [anchor=south] {} ;
\draw (a3) node [anchor=south] {$s_{2}$} ;
\draw (a4) node [anchor=north] {$s_{0}$} ;
\draw[-] (a1) -- node {} (a3);
\draw[-] (a1) -- node {} (a4);
\draw[-] (a3) -- node {} (a4);
\path[use as bounding box] (-1.5,0) rectangle (0,0);
\draw (1.4, 0.5) node  {$\curvearrowright$};
\draw (0.4, -0.9) node [rotate=135] {$\curvearrowright$};
\draw (2.4, -0.9) node [rotate=-135] {$\curvearrowright$};
\draw  (1.4, 0.3) node {$\rho$};
\end{tikzpicture}      
\caption{$\widetilde{\Ga}^{(1)}(A_2)$}\label{DaA21}
\end{figure}

Using the equations \eqn{de},  \eqn{transs} and \eqn{tat} we have the translational elements:
\begin{align}
&t_{\ot\al}=s_{\de-\ot\al}s_{\ot \al}=s_{\al_0}s_{\ot \al},\\
&t_{\al_1}=s_2s_{\al_0}s_{\ot \al}s_2,\\
&t_{\al_2}=s_1s_{\al_0}s_{\ot \al}s_1.
\end{align}
The actions of $t_{\ot\al}$, $t_{\al_1}$ and $t_{\al_2}$ on $V^{(1)}$ given 
 by Equation \eqn{Transl} imply the following shift motions on the affine simple roots :
\begin{align}
&t_{\ot\al} : (\al_0,\al_1, \al_2)\to(\al_0+2\de,\al_1-\de, \al_2-\de),\\
&t_{\al_1} : (\al_0,\al_1, \al_2)\to(\al_0+\de,\al_1-2\de, \al_2+\de),\\
&t_{\al_2} : (\al_0,\al_1, \al_2)\to(\al_0+\de,\al_1+\de, \al_2-2\de).
\end{align}
Observe that the above transformations satisfy the constraint \eqn{Tranc}.
The translational actions on the affine plane $h\in H$ is given by:
\beqn
t_{\al}. h= h+\oc \al,\quad  \al\in \{\ot \al,  \al_1, \al_2\}.
\eeqn
To describe translations correspond to
the fundamental weights $h_1, h_2$, we need the extended affine Weyl group  
$\widetilde{W}\bigl(A_2^{(1)}\bigr)=T_{ P} \rtimes W(A_2)=\langle t_{h_1},  t_{h_2}\rangle\rtimes W(A_2)$.
They are given by:
\beqn\label{Th1}
t_{h_1}=\rho s_2 s_1, \quad \mbox{and} \quad t_{h_2}=\rho^{-1} s_1 s_2.
\eeqn

Their actions result in shift motions on the simple affine roots given by:
\begin{align}\label{th1}
 &t_{h_1} :(\al_0,\al_1, \al_2)\to(\al_0+\de,\al_1-\de, \al_2),  \\\label{th2}
 &t_{h_2} :(\al_0,\al_1, \al_2)\to(\al_0+\de,\al_1, \al_2-\de),
\end{align}
and translational motions in the affine hyperplane $H$ are given by: 
\beqn
t_{h_i}. h= h+ h_i, \quad (i=1,2).
\eeqn
}

If we write $t_1=t_{h_1}$,
and let $\rho t_i=t_{i+1}\rho, (i=1,2,3)$, then we have:  
\begin{equation}\label{TaA2}
 t_1=\rho s_2s_1,\quad
 t_2=\rho s_0 s_2,\quad
 t_3=\rho s_1s_0,
\end{equation}
which correspond to the translations by: $\frac{1}{3}(2\e_1-\e_2-\e_3)$, $\frac{1}{3}(-\e_1+2\e_2-\e_3)$
and $\frac{1}{3}(-\e_1-\e_2+2\e_3)$, respectively.  They satisfy the relation
\beqn\label{TaA2p}
t_1t_2t_3=1.
\eeqn
The actions of $t_2$ and $t_3$ on the simple affine roots are given by:
\begin{align}\label{t2}
 &t_{2} :(\al_0,\al_1, \al_2)\to(\al_0,\al_1+\de, \al_2-\de),  \\\label{t3}
 &t_{3} :(\al_0,\al_1, \al_2)\to(\al_0-\de,\al_1, \al_2+\de).
\end{align}

We summarise the ideas explained in this section using Figure
\ref{a2V lattice}. It shows the affine plane $H$ on which $\widetilde{W}(A_2^{(1)})$ act as affine transformations by reflections about the reflection lines $H_{\al+k\de}$, and rotation about the center of the fundamental simplex $F$ by $\rho$.
The fundamental simplex $F$ is bounded by three red reflection lines 
$H_{\al_0}, H_{\al_1}$ and $H_{\al_2}$ that correspond to the three simple
reflections $s_0, s_1$ and $s_2$. 
It is
the triangle with vertices $U_0$, $U_{1}$, $U_{12}$, where $U_0=h_0$ is the origin $\mathbf 0$ and $U_1=h_1, U_{12}=h_2$ are the two
fundamental weights. 
The hexagon with vertices $U_{1}, U_{2}, U_{3}, U_{12}, U_{13}, U_{23}$ centered at $U_0$ is {\rm Vor}$(A_2)$.  
We have also shown the vertex $\oc {\ot \al}=U_{112}=U_{12}+U_{1}$.
Any point on the $A_2$ weight lattice can be obtained from the three vertices of $F$ (which we call the initial value set): $\{U_0$, $U_{1}$, $U_{12}\}$
using the elements of $\widetilde{W}(A_2^{(1)})$. The transformations on the the initial value set is given by
\begin{align}\nonumber
&s_1:\{U_1\lra U_2\},\\\label{nataA2}
&s_2:\{U_{12}\lra U_{13}\},\\\nonumber
&s_0:\{U_{0}\lra U_{112}\},\\\nonumber
&\rho :\;(U_0,U_1, U_{12})\mapsto (U_1,U_{12}, U_0),
\end{align}
which satisfy the defining relations of $\widetilde{W}(A_2^{(1)})$ given in Equation \eqn{funWaA2}.
We have also shown the translational actions of 
$t_{\al_1}$, $t_{\ot\al}$, $t_{1}, t_{2}$ and $t_{3}$ on $F$. 

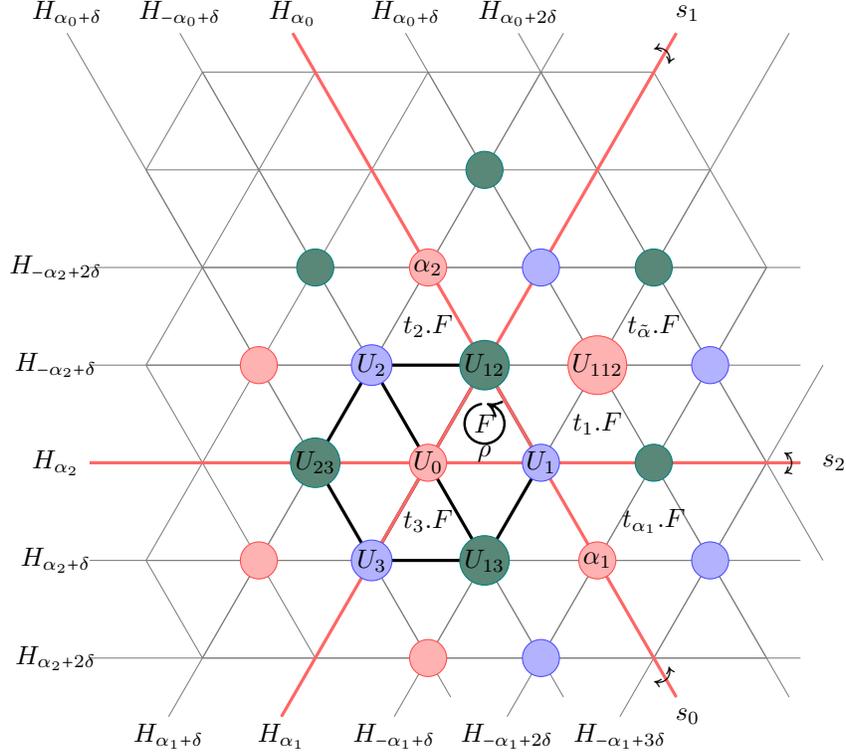
\begin{figure}[ht!] 
\centering
  \begin{tikzpicture}[scale=1.5, every annotation/.style={fill=white}
  ,every node/.style={circle,radius=0.7cm}]]


               \foreach \s in {2,...,6} \foreach \t in {-3,-1,1}
		{

\draw [gray] (0.5+\s,{(\t+2)*sqrt(3)/2}) -- (0+\s, {(\t+3)*sqrt(3)/2});
\draw  [gray] (-0.5+\s,{(\t+2)*sqrt(3)/2}) -- (0+\s, {(\t+3)*sqrt(3)/2});
\draw  [gray](-0.5+\s,{(\t+2)*sqrt(3)/2}) -- (0.5+\s, {(\t+2)*sqrt(3)/2});
                }      
                   \foreach \s in {2,...,6} \foreach \t in {-2,0,2}
		{

\draw  [gray](0+\s,{(\t+2)*sqrt(3)/2}) -- (-0.5+\s, {(\t+3)*sqrt(3)/2});
\draw [gray](-1+\s,{(\t+2)*sqrt(3)/2}) -- (-0.5+\s, {(\t+3)*sqrt(3)/2});
\draw  [gray](-1+\s,{(\t+2)*sqrt(3)/2}) -- (0+\s, {(\t+2)*sqrt(3)/2});
                }  
            
\draw  [gray](6,{sqrt(3)})--(6.5,{1.5*sqrt(3)});
\draw  [gray](6,0)--(6.5,{0.5*sqrt(3)});

\draw  [gray](1,0)--(1.5,-{0.5*sqrt(3)});
\draw  [gray] (1,{sqrt(3)})--(1.5,{0.5*sqrt(3)});
\draw  [gray] (1,{2*sqrt(3)})--(1.5,{1.5*sqrt(3)});
 \draw  [gray] (1.5,{2.5*sqrt(3)})--(5.5,{2.5*sqrt(3)});
 

 \foreach \s in {-1,0,1}
 { 
\draw  [gray]  (2.2+\s, -{0.8*sqrt(3)})-- (5.7+\s, {2.7*sqrt(3)});
 }   
 \draw  [gray]  (4.2, -{0.8*sqrt(3)})-- (6.5, {1.5*sqrt(3)});
 \draw   [gray]  (5.2, -{0.8*sqrt(3)})-- (7, {sqrt(3)});
 \foreach \s in {-1,0,...,3}
 { 
\draw  [gray] (0.5, {\s*0.5*sqrt(3)})--(6.8, {\s*0.5*sqrt(3)});
 }  
 \foreach \s in {-1,...,2}
 { 
\draw  [gray]   (4.7+\s, -{0.7*sqrt(3)})-- (1.3+\s, {2.7*sqrt(3)});
 }  
\draw  [gray]   (7, 0)-- (1.3+3, {2.7*sqrt(3)});

\foreach \s in {0.5} 
{
\draw[very thick,black] (2.5+\s, 0) -- (2+\s, {0.5*sqrt(3)}) ; 
\draw[very thick,black] (2+\s, {0.5*sqrt(3)}) -- (2.5+\s, {sqrt(3)})  ;
\draw[very thick,black] (2.5+\s, {sqrt(3)}) -- (3.5+\s, {sqrt(3)})  ;
\draw[very thick,black] (3.5+\s, {sqrt(3)}) -- (4+\s, {0.5*sqrt(3)});
\draw[very thick,black](4+\s, {0.5*sqrt(3)}) -- (3.5+\s, 0)  ;
\draw[very thick,black] (3.5+\s, 0) -- (2.5+\s, 0)  ;
\draw[very thick,black]   (2.5+\s, {sqrt(3)})--(3.5+\s, 0);
\draw[very thick,black]    (2.5+\s, 0)-- (3.5+\s, {sqrt(3)});

}
\draw[very thick, notred]    (2.2, -{0.8*sqrt(3)})-- (5.7, {2.7*sqrt(3)});
\draw[very thick,notred] (0.5, {0.5*sqrt(3)})--(6.8, {0.5*sqrt(3)});
\draw[very thick, notred]    (5.7, -{0.7*sqrt(3)})-- (2.3, {2.7*sqrt(3)});
\node [root] at (3.5, {0.5*sqrt(3)})  {};
\node  [leaf] at (4.5, {0.5*sqrt(3)})  {};
\node  [gr] at (4, {sqrt(3)}) {};
\node [root] at (2, {sqrt(3)}) {};
\node [root] at (5, {sqrt(3)}) {$U_{112}$};
\node [root] at (2, 0){};
\node [root] at (5, 0) {};
\node [root] at (3.5, {1.5*sqrt(3)})  {};
\node [root] at (3.5, {-0.5*sqrt(3)})  {};
\node [leaf] at (3, {sqrt(3)}){$U_{2}$};
\node  [leaf] at (4.5, {1.5*sqrt(3)}){};
\node  [leaf] at (4.5, {-0.5*sqrt(3)}){};
\node [leaf] at (3, 0){$U_{3}$};
\node [leaf] at (6, 0) {};
\node [leaf] at (6, {sqrt(3)}) {};
\node at (3.5, {0.5*sqrt(3)})  {$U_0$};
\node   at (4.5, {0.5*sqrt(3)})  {$U_1$};
\node at (5, 0) {$\al_1$};
\node  at (3.5, {1.5*sqrt(3)})  {$\al_2$};
\node [gr] at (2.5, {0.5*sqrt(3)}){$U_{23}$};
\node [gr] at (2.5, {1.5*sqrt(3)}){};
\node  [gr] at (4, {2*sqrt(3)}){};
\node  [gr] at (4, {sqrt(3)}) {$U_{12}$};
\node  [gr] at (4, 0){$U_{13}$};
\node [gr] at (5.5, {0.5*sqrt(3)}) {};
\node [gr] at (5.5, {1.5*sqrt(3)}) {};

\draw (0.2, {1.5*sqrt(3)}) node {$H_{-\al_2+2\de}$};
\draw (0.2, {sqrt(3)}) node {$H_{-\al_2+\de}$};
\draw (0.2, {0.5*sqrt(3)}) node {$H_{\al_2}$};
\draw (0.2, 0) node {$H_{\al_2+\de}$};
\draw (0.2, {-0.5*sqrt(3)}) node {$H_{\al_2+2\de}$};
\draw (7.1, {0.5*sqrt(3)}) node {$s_2$};
\draw (6.7, {0.5*sqrt(3)}) node [rotate=180] {$\curvearrowupdown$};

\draw (2.2, -{0.9*sqrt(3)}) node {$H_{\al_1}$};
\draw (1.2, -{0.9*sqrt(3)}) node {$H_{\al_1+\de}$};
\draw (3.2, -{0.9*sqrt(3)}) node {$H_{-\al_1+\de}$};
\draw (4.2, -{0.9*sqrt(3)}) node {$H_{-\al_1+2\de}$};
\draw (5.2, -{0.9*sqrt(3)}) node {$H_{-\al_1+3\de}$};
\draw (5.8, -{0.8*sqrt(3)}) node {$s_0$};
\draw (5.6, -{0.6*sqrt(3)}) node [rotate=135] {$\curvearrowupdown$};

\draw (1.3, {2.8*sqrt(3)}) node {$H_{-\al_0+\de}$};
\draw (2.3, {2.8*sqrt(3)}) node {$H_{\al_0}$};
\draw (5.8, {2.8*sqrt(3)}) node {$s_1$};
\draw (5.6, {2.6*sqrt(3)}) node [rotate=-135] {$\curvearrowupdown$};
\draw (3.3, {2.8*sqrt(3)}) node {$H_{\al_0+\de}$};
\draw (4.3, {2.8*sqrt(3)}) node {$H_{\al_0+2\de}$};
\draw (0.3, {2.8*sqrt(3)}) node {$H_{\al_0+\de}$};
\draw (4, {0.7*sqrt(3)}) node {$F$};
\draw (5, {0.7*sqrt(3)}) node {$t_{1}.F$};
\draw (3.5, {0.2*sqrt(3)}) node {$t_{3}.F$};
\draw (3.5, {1.2*sqrt(3)}) node {$t_{2}.F$};
\draw (5.5, {0.2*sqrt(3)}) node {$t_{\al_1}.F$};
\draw (5.5, {1.2*sqrt(3)}) node {$t_{\ot \al}.F$};
\draw (4, {0.7*sqrt(3)}) node {{\Huge$\circlearrowleft$}};
\draw (4, {0.55*sqrt(3)}) node {$\rho$};
\end{tikzpicture}
\caption{Translational actions of $\widetilde{W}(A_2^{(1)})$ on the affine plane $H$}\label{a2V lattice}
\end{figure}
Having introduced the preliminaries of the Weyl groups we are now ready to discuss them in the context of discrete
integrable systems.

\section{Discrete integrable systems on Voronoi cells -- a combinatorial construction}\label{DIS}
 An important class of discrete integrable equations, known as quad-equations, were
  classified by Adler-Bobenko-Suris \cite{abs:03}. 
 They considered multi-dimensionally consistent, affine linear 
 equations, which relate values of a function on vertices of a quadrilateral. These equations
 are the discrete analogues of many well-known nonlinear integrable partial differential equations.
 
 Our construction relies on reformulating systems of quad-equations as birational group
 representations of the Weyl groups. In particular, we associate the system of quad-equations
with a polytope of the Weyl group, where quad-equations are associated with some quadrilateral 
sub-structures of the polytope. The consistency and symmetry of the system of equations
are the consequences of the combinatorial structure and the symmetry of the polytope.
 
 To formulate discrete dynamical systems as birational groups is an approach well-established in 
 the studies of quadrirational Yang Baxter maps \cite{ves:07} and discrete \Pa equations \cite{NY:98, sak:01, KMNOY:06}.
 In the context of quad-equations, it was first used by
Atkinson \cite{akt:13} 
in his work on Yang Baxter maps as birational representation of a sequence of Coxeter
groups (associated with the connected T-shaped Coxeter-Dynkin diagram). Our construction is different
from those of  \cite{akt:13}, in that we allow different quad-equations in a system.
Consequently,  we can construct 
systems with symmetries associated with the disconnected Dynkin diagrams.

In order to associate equations and variables with the polytopes discussed in the previous section,  
 let us consider the function $x:  \{X_S\} \to \mb{C}$, on the vertices of {\rm Vor}$(B_n)$ ($n$-cube),
we write
\begin{align}
 &x_S=x(X_S), \label{xfun} 
\end{align}
where $S\subseteq\{1, ..., n\}$,
so that the function  $x$ is indexed by the weight vectors of the $B_n$ root system,
 and elements $w\in W(B_n)$ act on the variables $x_S$ by
 \begin{align}
 &w(x_S)=x(w.X_S)=x(X_{w.S})=x_{w.S}\label{wxfun}.
\end{align}

Similarly, consider the function $u: \{U_S\} \to \mb{C}$ on the vertices of Vor$(A_{n-1})$,
we write
\begin{align}
 &u_S=u(U_S), \label{ufun} 
\end{align}
where $S\subseteq\{1, ..., n\}$, $1\leq |S| \leq n-1$, and $u_0=u({\bf 0})$.
The function  $u$ is indexed by the weight vectors of the $A_{n-1}$ root system,
in particular the centre and the vertices of {\rm Vor}$(A_{n-1})$. Elements of 
$w\in \widetilde{W}\bigl(A_{n-1}^{(1)}\bigr)$ act on the variables $u_S$ by
 \begin{align}
 &w(u_S)= u(w. U_S)= u( U_{w.S})= u_{w.S}\label{wufun}.
\end{align}

In what follows we give as examples the construction of two systems of quad-equations with different
 combinatorial structures: one associated with the $3$-cube ($W(B_3)$ symmetry)
 and the other with the asymmetric
 $3$-cube ($W(B_2+A_1)$ symmetry),  and give their reductions to some discrete \Pa type equations.
 
 \subsection{A system of quad-equations with $W(B_3)$ symmetry}\label{DISB3}
 We construct a system of six l-mKdV equations consistent on {\rm Vor}$(B_{3})$ (a $3$-cube), 
 hence showing that the system has $W(B_3)$ symmetry.
 
 The l-mKdV equation is defined by :
 \beqn\label{f13}
 Q(x_0,x_1, x_3, x_{13} ;a_1,a_3)
=x_{1}x_{13}-x_0x_{3}+\frac{a_1}{a_3}(x_0x_{1}-x_{3}x_{13})=0,
\eeqn
where the variables $\{x_0,x_1, x_3, x_{13}\}$ are assigned to the vertices of 
the quadrilateral $\{X_0,X_1, X_3, X_{13}\}$, and $a_i$ are parameters associated with the edges.

{\prop [\cite{abs:03}] \label{abs:03}
A system of six such quad-equations, associated to the two-dimensional faces (2-faces) of the $3$-cube (see Figure \subref{3cube}),
\begin{subequations}\label{f}
\begin{align}\label{ff13}
&Q(x_0,x_{1},x_{3},x_{13};a_1, a_3)=0,\\\label{f12}
&Q(x_0,x_1, x_2, x_{12} ;a_1,a_2)=0,\\\label{f23}
&Q(x_0,x_2, x_3, x_{23} ;a_2,a_3)=0,\\\label{f123}
&Q(x_{3},x_{13}, x_{23}, x_{123} ;a_1,a_2)=0,\\\label{f132}
&Q(x_2,x_{12}, x_{23}, x_{123} ;a_1,a_3)=0,\\\label{f231}
&Q(x_1, x_{13}, x_{12}, x_{123} ;a_3,a_2)=0,
\end{align}
\end{subequations}
is said to be consistent on the $3$-cube if given four initial values $\{x_0,x_1, x_2, x_{3}\}$ (indicated by
black nodes), if
$x_{123}$ can be uniquely and consistently determined from the equations \eqref{f123}--\eqref{f231}.
}                                                                                         

The natural action of $W(B_3)$ on the index set $S$ of $X_S$ (given by Equation \eqn{natB3})
is extended to that on the variables $x_S$ by Equation \eqn{wxfun}. The equations in \eqn{f} can be
obtained from Equation \eqn{f13} by applying the following elements $\{1, s_2, s_1, s_3s_2, s_2s_3s_2, s_1s_2s_3s_2\}$
of $W(B_3)$, respectively. 

\begin{figure}\label{Fb3a2}
\begin{subfigure}[b]{.4\textwidth}
\hs{-0.5cm}
\resizebox{\linewidth}{!}{
\begin{tikzpicture}[scale=0.4]	
\filldraw [black] (0,0) node [anchor=north east] {$x_3$}  circle(1ex);
\draw [black]  (0,5) node [anchor=east] {$x_{23}$ } ;
\draw [black] (5,0) node [anchor=north west] {$x_{13}$ } ;
\draw [black](5,5) node [anchor=west] {$x_{123}$ } ;

\filldraw [black]  (2,2) node [anchor=east] {$x_0$} circle(1ex);
\filldraw [black] (2,7) node [anchor=east] {$x_2$} circle(1ex);
\filldraw [black]  (7,2) node [anchor=west]{$x_1$} circle(1ex);
\draw [black]  (7,7) node [anchor=west] {$x_{12}$ };

\draw[very thick, loosely dotted] (2,2)--(5,5);

\draw [thick] (0,0) -- (0,5);
\draw [thick] (0,0) -- (5,0);
\draw [thick] (5,0) -- (5,5);
\draw [thick] (0,5) -- (5,5);

\draw [thick] (5,0) -- (7,2);
\draw [thick] (5,5) -- (7,7);
\draw [thick] (0,5) -- (2,7);
\draw [thick] (2,7) -- (7,7);
\draw [thick] (7,2) -- (7,7);

\draw [thick](0,0) -- (2,2);
\draw [thick](2,2) -- (2,7);
\draw [thick] (2,2) -- (7,2);

\end{tikzpicture}
}
\subcaption{A system of quad-equations with initial values $\{x_0,x_1, x_2, x_{3}\}$.}
\label{3cube}
\end{subfigure}
\begin{subfigure}[b]{.5\textwidth}
\begin{center}
\begin{tikzpicture}[scale=1.5]
 		\foreach \s in {2,3,4,5}
		{
\draw[dashed] (0+\s,0) -- (0.5+\s, 0.866025403784);
\draw[dashed] (1+\s, 0.866025403784) -- (-1+\s, 0.866025403784);
\draw[dashed] (-0.5+\s, 0.866025403784) -- (-1+\s,0) ;
\draw[dashed] (-1+\s,0) -- (\s,0);
\draw[dashed] (-0.5+\s, -0.866025403784) -- (0.5+\s, -0.866025403784) ;
\draw[dashed] (-0.5+\s, 0.866025403784)--(0+\s, 0);

               }

\foreach \s in {2,3,4}
		{
\draw[dashed] (\s,0) -- (0.5+\s, -0.866025403784);
\draw[dashed] (\s,0) -- (-0.5+\s, -0.866025403784);
\draw[dashed] (0.5+\s, -0.866025403784) -- (1+\s,0);
\draw[dashed] (1+\s,1.73205080757) -- (0.25+\s, 2.89807621135);
\draw[dashed] (0+\s,1.73205080757) -- (0.5+\s, 2.59807621135);
\draw[dashed] (0.5+\s,2.59807621135) -- (0.75+\s, 3.03105);
\draw[dashed] (\s,1.73205080757) -- (-0.5+\s, 0.866025403784);
		}

 		\foreach \s in {2,3,4,5}
		{

\draw[dashed] (0.5+\s, 2.59807621135) -- (-0.5+\s, 2.59807621135);
\draw[dashed] (-1+\s,1.73205080757) -- (\s,1.73205080757);
\draw[dashed] (-0.5+\s, 0.866025403784) -- (0.5+\s, 0.866025403784) ;
\draw[dashed] (\s,1.73205080757) -- (0.5+\s, 0.866025403784);
\draw[dashed] (0.5+\s, 0.866025403784) -- (1+\s,1.73205080757);
		}
\draw[thick] (3, 0) -- (2.5, 0.866025403784) ; 
\draw[thick] (2.5, 0.866025403784) -- (3, 1.73205080757)  ;
\draw[thick] (3, 1.73205080757) -- (4, 1.73205080757)  ;
\draw[thick] (4, 1.73205080757) -- (4.5, 0.866025403784);
\draw[thick] (4.5, 0.866025403784) -- (4, 0)  ;
\draw[thick] (4, 0) -- (3, 0)  ;

\filldraw [black]  (6.2, 0.866025403784) node   {\colorbox{white}{$s_2$}} circle(0.1ex);
\filldraw [black]  (4.75, 3.13205080757) node  {\colorbox{white}{$s_1$}} circle(0.1ex);
\filldraw [black]  (3, 3.13205080757) node  {\colorbox{white}{$s_0$}} circle(0.1ex);

\filldraw [black]  (4, 0.53) node  {\colorbox{white}{{\footnotesize$s_2.F$}}} circle(0.1ex);
\filldraw [black]  (4, 1.23) node  {\colorbox{white}{{\footnotesize$F$}}} circle(0.1ex);
\filldraw [black]  (3.5, 0.33) node  {\colorbox{white}{{\footnotesize$s_2s_1.F$}}} circle(0.1ex);
\filldraw [black]  (3.5, 1.53) node  {\colorbox{white}{\footnotesize{$s_1.F$}}} circle(0.1ex);
\filldraw [black]  (3, 0.63) node  {\colorbox{white}{{\footnotesize$s_2s_1s_2.F$}}} circle(0.1ex);
\filldraw [black]  (3, 1.03) node  {\colorbox{white}{{\footnotesize$s_1s_2.F$}}} circle(0.1ex);
\filldraw [black]  (5, 1.23) node  {\colorbox{white}{{\footnotesize$t_1.F$}}} circle(0.1ex);

\filldraw [black]   (3.5, 0.866025403784) node [mybox] (box) {\bf $u_0$} circle(0.1ex);
\filldraw [black]  (4.5, 0.866025403784) node   [mybox] (box) {\bf $u_1$} circle(0.1ex);
\filldraw [black]  (3, 1.73205080757) node  {\colorbox{white}{$u_2$}} circle(0.1ex);
\filldraw [black]  (3, 0) node  {\colorbox{white}{$u_3$}} circle(0.1ex);
\filldraw [black]  (2.5, 0.866025403784) node  {\colorbox{white}{$u_{23}$}} circle(0.1ex);
\filldraw [black]  (4, 1.73205080757) node   [mybox] (box) {\bf $u_{12}$} circle(0.1ex);
\filldraw [black]  (4, 0) node  {\colorbox{white}{$u_{13}$}} circle(0.1ex);
\filldraw [black]  (5, 1.73205080757) node  {\colorbox{white}{$u_{112}$}} circle(0.1ex);
\filldraw [black]  (5.5, 0.866025403784) node  {\colorbox{white}{$u_{11}$}} circle(0.1ex);
\draw (5.9, 0.866025403784) node [rotate=-180] {$\curvearrowupdown$};
\draw (3.27, 2.83205080757) node [rotate=-45] {$\curvearrowupdown$};
\draw (4.66, 2.83205080757) node [rotate=-135] {$\curvearrowupdown$};
\end{tikzpicture}
\end{center}
\subcaption{}
\label{cth}
\end{subfigure}
\end{figure}
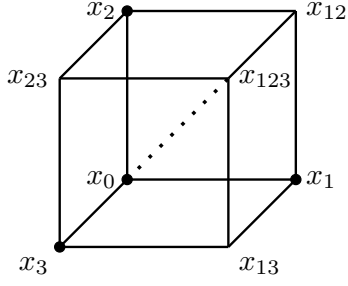

The quad-equation \eqn{f13} being affine linear means
that any one of the four variables can be expressed rationally in terms of the other three, given as initial values. 
For example we can write
\;
\beqn\label{x13}
x_{13}=\frac{x_0(-a_1 x_1+a_3 x_3)}{a_3 x_1-a_1x_3}.
\eeqn
Similarly, we can express $x_{12}, x_{23}$ and $x_{123}$ all in terms of $\{x_{0},x_{1},x_{2},x_{3}\}$. 

Let $\mb{C}(a; x)$ be the field of rational functions 
in $a_j (j\in\{1,2,3\})$ and $x_i (i\in I=\{0,1,2,3\})$. 
 In order to obtain a representation
of $W(B_3)$ on the initial value set $x_i (i\in I=\{0,1,2,3\})$ we need some appropriate 
algebraic relations among the variables $x_S$ which are associated with the
vertices of the $3$-cube. To this end, we use the quad-equations associated with the $2$-faces 
of the cube such as expression \eqn{x13}. The whole system \eqn {f} involving eight variables then can be reformulated as a rational system of 
initial value set
$\{x_{0},x_{1},x_{2},x_{3}\}$, induced from the actions of the symmetry group of the $3$-cube, that is $W(B_3)$.

{\Def \label{DefactB3}
The actions on the variables $x_S$ and $a_i$ to be associated with the generators of $W(B_3)$ on the initial value set
$\{x_{0},x_{1},x_{2},x_{3}\}$
are defined as:
\begin{align}\nonumber
&s_1:\{x_1\lra x_2,\quad a_1\lra a_2\},\\\label{actioneB3o2}
&s_2:\{x_2\lra x_3,\quad a_2\lra a_3\},\\\nonumber
&s_3:\{x_0\lra x_3,\quad
x_{1}\to \frac{x_0(-a_1 x_1+a_3 x_3)}{a_3 x_1-a_1x_3},\quad 
x_{2}\to \frac{x_0(-a_2 x_2+a_3 x_3)}{a_3 x_2-a_2x_3},\\\nonumber
&\quad a_1\to -a_1,\quad a_2\to -a_2\}.
\end{align}
}

{\prop \label{sB3}
The associated actions of $s_1$, $s_2$ and $s_3$ defined as above satisfy the fundamental relations in \eqn{funWB3} 
for the simple reflections of $W(B_3)$, corresponding to the Dynkin diagram in Figure \subref{DB3}.
That is, the actions of these generators define a birational representation of $W(B_3)$ on the field $\mb{C}(a; x)$ of rational
functions. Furthermore, system \eqn{actioneB3o2} is equivalent to the system of quad-equations \eqn{f}.
}
\begin{proof}
It can be checked by direct computation that transformations in \eqn{actioneB3o2} satisfy the fundamental relations in \eqn{funWB3}.  
\end{proof}
In what follows, we give a reduction of system \eqn{f} to
a discrete equation of \Pa type using the relation between {\rm Vor}$(B_3)$ and {\rm Vor}$(A_{2})$
given in Example \ref{B3}.

\subsection{Reduction to a $q$-discrete \Pa equation with $\widetilde{W}\bigl(A_2^{(1)}\bigr)$ symmetry}\label{DISA2}

{\prop
 On system \eqn{f} apply the map $\phi$  
 associated with the orthogonal projection $\phi$ define in \eqn{phi}
 along $X_{123}=\frac{1}{2}(\e_1+\e_2+\e_3)$,
 \beqn\label{pc}
 \begin{cases}
 &\phi\,(x_S)=u_S,\quad 1\leq |S| \leq2,\\
 &u_{123}=\phi\, (x_{123})=\phi\, (x_0)=u_0,
 \end{cases}
\eeqn
we have the ``periodic condition'' (indicated by the dotted line in Figure \subref{3cube}):
\beqn\label{conp}
t_{1}t_{2}t_{3}(u_S)=u_{S\cup \{1, 2, 3\}}=u_{S},
\eeqn
which puts the following constraints the edge parameters $a_i$:
\beqn\label{conpp}
t_{j}(a_i)=
\begin{cases}
a_i&\text{for $j\neq i$}\\
qa_i&\text{for $j= i$}
\end{cases},
\eeqn
where $q\in\mathbb{C}^\ast$ and $t_{j}(u_S)=u_{S\cup \{j\}}$ (defined in \eqn{TaA2})
is the one-step shift in direction $j$ on the $A_2$ weight lattice (see Figure \ref{a2V lattice}).
System \eqn{f} is reduced to 

\begin{subequations}\label{actioneA2o1}
\bee\label{w112}
&&u_{112}=\frac{u_{0}(q a_1 u_1+a_3 u_{12})}{a_3 u_1+q a_1 u_{12}},\\\label{w12}
&&u_{2}=\frac{u_{1}(a_1 u_0+a_2 u_{12})}{a_2 u_0+a_1 u_{12}},\\\label{w13}
&&u_{13}=\frac{u_{12}(a_1 u_1+a_3 u_{0})}{a_3 u_1+a_1 u_{0}}.
\ee
\end{subequations}
}

\begin{proof} 
It can be easily checked that condition \eqn{pc} on system \eqn{f} leaves constraints 
\eqn{conp}--\eqn{conpp} on the 
variables and parameters, respectively. Three essential quad-equations 
\eqn{w112}--\eqn{w13} describe the reduced system.
For $\ohll u$, we shift \eqn{ff13} by $t_1t_2=t_{12}$; Equations \eqn {w12} and \eqn{w13}
 come from Equations \eqn{f12} and \eqn{f231}, respectively. The rest of 
 the equations of the reduced system \eqn{f} can be obtained from these three equations by
 shifts under the constraints given in Equations \eqn{conp}--\eqn{conpp}. 
  \end{proof}
The natural action of $\widetilde{W}\bigl(A_2^{(1)}\bigr)$ on the index set $S$ of $U_S$ (given by Equation \eqn{nataA2})
is extended to that on the variables $u_S$ by Equation \eqn{wufun}. The three equations in \eqn{actioneA2o1} are
associated with the three quadrilaterals around the fundamental simplex $F$ (see Figure \subref{cth}). These are used to obtain a representation
of ${W}\bigl(A_2^{(1)}\bigr)$ on the initial value set $\{u_0,u_1,u_{12}\}$, which are the
vertices of $F$, indicated by the rounded boxes in Figure \subref{a1a1}.
{\prop \label{saA2}
The reduced system \eqn{actioneA2o1} can be formulated as a birational
representation of $\widetilde{W}\bigl(A_2^{(1)}\bigr)=\langle s_0, s_1, s_2, \rho\rangle$
on the three initial values $\{u_0, u_{1}, u_{12}\}$ 
by defining the actions associated to the generators as follows,

\begin{align}\nonumber
&s_1:\{u_1\to u_2=\frac{u_{1}(a_1u_0+a_2u_{12})}{a_2u_0+a_1u_{12}},\quad a_1\lra a_2\}, \\\label{actioneA2o2}
&s_2:\{ u_{12}\to u_{13}=\frac{u_{12}(a_2u_1+a_3 u_{0})}{a_3 u_1+a_1u_{0}},\quad a_2\lra a_3\}, \\\nonumber
&s_0:\{u_0\to u_{112}=\frac{u_{0}(q a_1u_1+a_3 u_{12})}{a_3 u_1+q a_1u_{12}},\quad
a_3\to qa_1, \quad a_1\to a_3/q\},\\\nonumber
&\rho:\{u_0\to u_1,\quad u_1\to u_{12}, \quad u_{12}\to u_0,\quad a_1\to a_2,\quad a_2\to a_3, \quad
a_3\to q a_1\}.
\end{align}

The associated actions of $s_1, s_2, s_0, \rho$ defined as above satisfy the fundamental relations \eqn{funWaA2}
of $\widetilde{W}\bigl(A_2^{(1)}\bigr)$ (corresponding to the Dynkin diagram in Figure \ref{DaA21}).
We note that $\rho^3:\{a_1\to q a_1,\,a_2\to q a_2, \,
a_3\to q a_3\}$, but since the edge parameters $a_i$ of system \eqn{actioneA2o1}
can be always arranged as pairs of ratios and $\rho^3:\{\frac{a_i}{a_j}\to \frac{a_i}{a_j}\}$ for all $i, j\in \{1, 2, 3\}$,
the defining relation $\rho^3=1$ in \eqn{funWaA2} is satisfied.
}
\begin{proof} 
It can be verified directly that
 the transformations given in Equation \eqn{actioneA2o2} satisfy the relations in Equation \eqn{funWaA2}. That is the reduced System \eqn{actioneA2o1}
 has $\widetilde{W}\bigl(A_2^{(1)}\bigr)$ symmetry.
\end{proof}

{\Rem
The dynamics of system \eqn{actioneA2o1} is described by the translational elements
of $\widetilde{W}\bigl(A_2^{(1)}\bigr)=T_{ P} \rtimes W(A_2)=\langle t_{h_1},  t_{h_2}\rangle\rtimes W(A_2)$.
Recall that $t_1=t_{h_1}$,
where $h_1=U_1=(2\e_1-\e_2-\e_3)/3$ is a fundamental weight of $W(A_2)$, and $t_1$, $t_2$, $t_3$ defined in
Equation \eqn{TaA2} satisfy the condition \eqn{TaA2p}, which
is the ``periodic condition'' \eqn{conp}.
Furthermore, their actions on  $\frac{a_i}{a_j}$ for all $i, j\in \{1, 2, 3\}$, given by Equation \eqn{conpp},  satisfy the condition \eqn{TaA2p}.                                                                                                                                                                                                                                                                                                                                                                              
}

 {\theorem
Let 
\beqn\label{fg}
f=\frac{\ol u}{\ohl u},\;\; g=\frac{\ohl u}{u_0}.
\eeqn

System \eqn{actioneA2o1} is equivalent to 

\beqn\label{qPIII}
\ol g=\frac{(1+ft)}{fg(f+t)},\;\;
\ol f=\frac{(1+a\ol g t)}{f\ol g (\ol g+at)},\quad 
\eeqn
where $t=qa_1/a_3,\;\;a=a_3/a_2$, and $t_1(f)=\ol f$, $t_1(g)=\ol g$. 
System \eqn{qPIII} is a second-order nonlinear ordinary $q$-discrete equation
of \Pa type with $\widetilde{W}\bigl(A_2^{(1)}\bigr)$ symmetry.
}

\begin{proof}

We want to find $f$, $g$ shifted in the $t_1$
direction, that is:

\begin{subequations}\label{T1fg}
\begin{align}
 &t_1(f)=\ol f={\oll u}/{\ohll u},  \\
 &t_1(g)=\ol g=\,{\ohll u}/{\ol u}.
\end{align}
\end{subequations}
For $\oll u$, shift Equation \eqn{w12} by $t_1$, rearrange,  we have
\beqn\label{w11}
\oll u=\frac{a_2\ol u\ohl u+qa_1\ohl u\ohll u}{qa_1\ol u+a_2\ohll u}.
\eeqn
We
use Equation \eqn{w112} to express $u_{112}$ 
in terms of $\{u_0, u_{1}, u_{12}\}$, and finally in terms of
of $f$ and $g$. We have system \eqn{qPIII},
where $t$ is the independent variable and $a$ is a parameter of the equation. 
\end{proof}

System \eqn{qPIII} is a sub-case of the $q$-discrete \Pa equation \eqn{eqn:qp3} with $\widetilde{W}\bigl((A_2+A_1)^{(1)}\bigr)$ symmetry,
which has been obtained from a similar type of reduction of a system of quad-equations \cite{jns1}.
In Section \ref{ag}, we provide an algebro-geometric description of Equation \eqn{eqn:qp3}.

\subsection{A system of quad-equations with $W(B_2+A_1)$ symmetry}\label{DISB2A1}
In the previous section we have constructed system \eqn{f} of six l-mKdV equations on the vertices of a $3$-cube
with $W(B_3)$ symmetry. In \cite{Boll:11c}, it was shown that different quad-equations can be
defined consistently on a $3$-cube,
sometimes referred to as the system on an asymmetric $3$-cube.
{\prop [\cite{abs:03}]
Given a $H^4$ type:
\begin{subequations}\label{eqna1a2}
\begin{equation}\label{a23}
 Q(x_0, x_2, x_3, x_{23} ;\be,\ga)
=x_{2}x_{23}+x_0x_{3}-\frac{\ga}{\be}(x_{3}x_{23}+x_0x_{2})=0, 
\end{equation}
and a $H^6$ type:
\begin{equation} \label{a12}
 H(x_0, x_1, x_2, x_{12} ;\rho,\be)
=\be x_{2}x_{12}+\be x_0x_{1}+ \rho x_{0}x_{2}=0,
\end{equation}
\end{subequations}
quad-equations in the ABS classification \cite{abs:03},
the following system of quad-equations:
\noindent
\begin{subequations}\label{af}
\begin{center}

\begin{minipage}[t]{.45\textwidth}
\begin{align}\label{af23}
\frac{x_{3}}{x_2}&=\frac{\ga x_0-\be x_{23}}{\be x_0-\ga x_{23}},  \\\label{af12}
\frac{x_1}{x_2}&=-\frac{x_{12}}{x_0}-\frac{\rho}{\be}, \\\label{af123}
\frac{x_{13}}{x_{23}}&=-\frac{x_{123}}{x_3}-\frac{\rho}{\be},
\end{align}
\end{minipage}
\begin{minipage}[t]{.45\textwidth}
\begin{align}\label{af231}
\frac{x_{13}}{x_{12}}&=\frac{\ga x_1-\be x_{123}}{\be x_1-\ga x_{123}}, \\\label{af13}
\frac{x_1}{x_3}&=-\frac{x_{13}}{x_0}-\frac{\rho}{\ga},\\\label{af132}
\frac{x_{12}}{x_{23}}&=-\frac{x_{123}}{x_2}-\frac{\rho}{\ga},
\end{align}
\end{minipage}
\end{center}
\end{subequations}
is consistent on a $3$-cube. Equations \eqn{af23} and \eqn{af231}) are of $H^4$ type, 
and Equations \eqn{af12}, \eqn{af123}, \eqn{af13},
and \eqn{af132} are of $H^6$ type. 
}
{\Rem
System \eqn{af} is equivalent to Equations (3.29-3.30) given in \cite{Boll:11c}  
with $\de_2=\de_3=0, \;\de_1=\rho$, and $x_0\leftrightarrow x_1,
 x_2\leftrightarrow x_{12},
 x_3\leftrightarrow x_{13},
  x_{23}\leftrightarrow x_{123},
    \be\rightarrow 1/\be,
\ga\rightarrow 1/\ga$. System \eqn{af} no longer have $W(B_3)$ symmetry of the $3$-cube. 
In fact we will prove that it has
$W(B_2+A_1)$ symmetry. 
}

An asymmetric cube (drawn by solid black lines in Figure \subref{a3cube}) can be constructed by considering its set of vertices $X_S$ ($S\subseteq\{1, ..., 3\}$) associated
with a set of edges of two adjacent
cubes (drawn by dashed lines in blue).
The symmetry group of such an asymmetric $3$-cube is $W(B_2+A_1)=\langle  s_2, r_{3},w_0\rangle$, where the
generators satisfy the following relations:
\beqn\label{actB2A1}
s_3^2=s_2^2=w_0^2=1,\quad
(s_2s_3)^4=1,\quad (s_2w_0)^2=(s_3w_0)^2=1.
\eeqn
The corresponding Dynkin diagram is given in Figure \subref{DB2A1}.

The natural actions of $s_2$, $s_3$ and $w_0$ on $X_S$ are given by the following transformations:
\begin{align}\nonumber
&s_2:\{X_2\lra X_3,\quad X_{12}\lra X_{13}\},\\\label{actac}
&s_3:\{X_0\lra X_3,\quad X_{1}\lra X_{13},\quad X_{2}\lra X_{23},\quad X_{12}\lra X_{123}\},\\\nonumber
& w_0:\{X_0\to 1/X_1,\quad X_1\to 1/X_0,
  \quad X_2\to 1/X_{12}, \quad X_{12}\to 1/X_{2},
   \quad X_3\to 1/X_{13},  \\\nonumber
 & \hs{1cm}  X_{13}\to 1/X_{3},
    \quad X_{23}\to 1/X_{123},\quad X_{123}\to 1/X_{23}\}.
\end{align}
The actions of $s_2$ and $s_3$ above are the same as those of $W(B_3)$ given in Equation \eqn{natB3}. The reflection plane associated with 
the generator $w_0$
is shown as the striped plane in Figure \subref{a3cube}.
It can be easily checked that the transformations in \eqn{actac} satisfy the fundamental relations \eqn{actB2A1} of $W(B_2+A_1)$.

\begin{figure}\label{Fb2a1}
\begin{subfigure}[b]{.5\textwidth}
\centering
\resizebox{\linewidth}{!}{
\begin{tikzpicture}[scale=0.4]		
%
\draw[dashed, blue, very thick] (-2.5,0) -- (-2.5,5);
\draw[blue, very thick] (2.5,0) -- (2.5,5);
\draw[dashed, blue, very thick] (7.5,0) -- (7.5,5);

\draw[dashed, blue, very thick] (-0.5,2)--(-0.5,7);
\draw[dashed, blue, very thick] (4.5,2)--(4.5,7);
\draw[dashed, blue, very thick] (9.5,2)--(9.5,7);

\draw[dashed, blue, very thick] (-2.5,0) --(-0.5,2);
\draw[dashed, blue, very thick] (2.5,0) --(4.5,2);
\draw[dashed, blue, very thick] (7.5,0) --(9.5,2);

\draw[dashed, blue, very thick] (-2.5,5) --(-0.5,7);
\draw[blue, very thick] (2.5,5) --(4.5,7);
\draw[dashed, blue, very thick] (7.5,5) --(9.5,7);

\draw[dashed, blue, very thick] (-2.5,0) -- (7.5,0);
\draw[dashed, blue, very thick] (-0.5,2)--(9.5,2);
\draw[dashed, blue, very thick] (-2.5,5)--(7.5,5);
\draw[dashed, blue, very thick] (-0.5,7)--(9.5,7);

\filldraw[blue!50, opacity=0.15] (-0.5,7)--(9.5,7)--(7.5,5)--(-2.5,5)--(-0.5,7);
\draw [black] (0,0) node [anchor=north east] {$x_3$} ;
\filldraw [black]  (0,5) node [anchor=north east] {$x_{23}$ } circle(1ex);
\draw [black] (5,0) node [anchor=north west] {$x_{13}$ } ;
\draw [black](5,5) node [anchor=north west] {$x_{123}$ } ;

\filldraw [black]  (2,2) node [anchor=north east] {$x_0$} circle(1ex);
\filldraw [black] (2,7) node [anchor=south east] {$x_2$} circle(1ex);
\draw [black](7,2) node [anchor=north] {$x_1$} ;
\filldraw [black](7,7) node [anchor=south west] {$x_{12}$ } circle(1ex);

\draw[very thick, loosely dotted] (2,2)--(5,5);
\foreach \s in {2.25}{
\filldraw [black]  (-1+\s,-2) node [anchor=east] {$w_0$ } circle(.1ex);
\filldraw [gray]  (-1+\s,7) node [anchor=east] {${}$ } circle(.1ex);
\filldraw [gray]  (3+\s,2) node [anchor=east] {${}$ } circle(.1ex);
\filldraw [black]  (3+\s,11) node [anchor=east] {$w_0$ } circle(.1ex);
\draw[pattern=horizontal lines, opacity=0.4] (-1+\s,-2)-- (-1+\s,7)--(3+\s,11)--(3+\s,2)--(-1+\s,-2);}
\draw [thick, black] (0,0) -- (0,5);
\draw [thick, black] (0,0) -- (5,0);
\draw [thick, black] (5,0) -- (5,5);
\draw [thick, black] (0,5) -- (5,5);

\draw [thick, black] (5,0) -- (7,2);
\draw [thick, black] (5,5) -- (7,7);
\draw [thick, black] (0,5) -- (2,7);
\draw [thick, black] (2,7) -- (7,7);
\draw [thick, black] (7,2) -- (7,7);

\draw [thick, black](0,0) -- (2,2);
\draw [thick, black](2,2) -- (2,7);
\draw [thick, black] (2,2) -- (7,2);
\draw [densely dotted] (5,0) -- (0,5);
\draw [densely dotted](5,0) -- (2,2);
\draw [densely dotted](5,0) -- (7,7);
\draw [densely dotted](0,5) -- (7,7);
\draw [densely dotted](0,5) -- (2,2);
\end{tikzpicture}
}
\subcaption{}\label{a3cube}
\end{subfigure}
\begin{subfigure}[b]{0.8\textwidth}
\centering
 \resizebox{\linewidth}{!}{
\begin{tikzpicture}[scale=1.4]


               \foreach \s in {2,...,5} \foreach \t in {-3,-1, 1}
		{

\draw[dashed, gray] (0.5+\s,{(\t+2)*sqrt(3)/2}) -- (0+\s, {(\t+3)*sqrt(3)/2});
\draw[dashed, gray] (-0.5+\s,{(\t+2)*sqrt(3)/2}) -- (0+\s, {(\t+3)*sqrt(3)/2});
\draw[dashed, gray] (-0.5+\s,{(\t+2)*sqrt(3)/2}) -- (0.5+\s, {(\t+2)*sqrt(3)/2});
                }      
                   \foreach \s in {2,...,6} \foreach \t in {-2,0}
		{

\draw[dashed, gray] (0+\s,{(\t+2)*sqrt(3)/2}) -- (-0.5+\s, {(\t+3)*sqrt(3)/2});
\draw[dashed, gray] (-1+\s,{(\t+2)*sqrt(3)/2}) -- (-0.5+\s, {(\t+3)*sqrt(3)/2});
\draw[dashed, gray] (-1+\s,{(\t+2)*sqrt(3)/2}) -- (0+\s, {(\t+2)*sqrt(3)/2});
                }  
            
\draw[dashed, gray] (5.5,{0.5*sqrt(3)})--(6,{sqrt(3)});
\draw[dashed, gray] (5.5,{-0.5*sqrt(3)})--(6,0);
\draw[dashed, gray] (1,0)--(1.5,-{0.5*sqrt(3)});
\draw[dashed, gray] (1,{sqrt(3)})--(1.5,{0.5*sqrt(3)});
\draw[dashed, gray] (1,{2*sqrt(3)})--(1.5,{1.5*sqrt(3)});
 \draw[dashed, gray] (1.5,{2*sqrt(3)})--(5.5,{2*sqrt(3)});
 
 
                \foreach \s in {2,...,5} \foreach \t in {-3,-1,1}
		{

\draw[thick, black] (1+\s,{(\t+2)*sqrt(3)/2}) -- (0.5+\s, {(\t+3)*sqrt(3)/2});
\draw[thick, black] (0+\s,{(\t+2)*sqrt(3)/2}) -- (0.5+\s, {(\t+3)*sqrt(3)/2});
\draw[thick, black] (0+\s,{(\t+2)*sqrt(3)/2}) -- (1+\s, {(\t+2)*sqrt(3)/2});
                }      
                   \foreach \s in {2,...,5} \foreach \t in {-2,0}
		{

\draw[thick, black] (0.5+\s,{(\t+2)*sqrt(3)/2}) -- (0+\s, {(\t+3)*sqrt(3)/2});
\draw[thick, black] (-0.5+\s,{(\t+2)*sqrt(3)/2}) -- (0+\s, {(\t+3)*sqrt(3)/2});
\draw[thick, black] (-0.5+\s,{(\t+2)*sqrt(3)/2}) -- (0.5+\s, {(\t+2)*sqrt(3)/2});
                }  
            
\draw[thick, black] (5.5,{sqrt(3)})--(6,{1.5*sqrt(3)});
\draw[thick, black] (5.5,0)--(6,{0.5*sqrt(3)});

\draw[thick, black] (1.5,0)--(2,-{0.5*sqrt(3)});
\draw[thick, black] (1.5,{sqrt(3)})--(2,{0.5*sqrt(3)});
\draw[thick, black] (1.5,{2*sqrt(3)})--(2,{1.5*sqrt(3)});
 \draw[thick, black] (1.5,{2*sqrt(3)})--(5.5,{2*sqrt(3)});
 
\foreach \s in {0.5} 
{
\draw[thick] (3+\s, 0) -- (2.5+\s, {0.5*sqrt(3)}) ; 
\draw[thick] (2.5+\s, {0.5*sqrt(3)}) -- (3+\s, {sqrt(3)})  ;
\draw[very thick] (3+\s, {sqrt(3)}) -- (4+\s, {sqrt(3)})  ;
\draw[thick] (4+\s, {sqrt(3)}) -- (4.5+\s, {0.5*sqrt(3)});
\draw[thick] (4.5+\s, {0.5*sqrt(3)}) -- (4+\s, 0)  ;
\draw[thick] (4+\s, 0) -- (3+\s, 0)  ;
\draw[thick]   (3+\s, {sqrt(3)})--(4+\s, 0);
\draw[thick]    (3+\s, 0)-- (4+\s, {sqrt(3)});

}

\draw[thick] (3.5, {-0.8*sqrt(3)})--(3.5, {2.2*sqrt(3)});
\draw[thick] (4, {-0.8*sqrt(3)})--(4, {2.2*sqrt(3)});
\draw[thick] (0.5, {0.5*sqrt(3)})--(6.5, {0.5*sqrt(3)});
\draw[thick] (0.5, {sqrt(3)})--(6.5, {sqrt(3)});

\draw[blue, dashed, very thick] (3, 0) -- (2.5, {0.5*sqrt(3)}) ; 
\draw[blue, very thick] (4, 0) -- (3.5, {0.5*sqrt(3)}) ; 
\draw[blue, dashed, very thick] (5, 0) -- (4.5, {0.5*sqrt(3)}) ; 
\draw[blue, dashed, very thick] (2.5, {0.5*sqrt(3)}) -- (3, {sqrt(3)})  ;
\draw[blue, very thick] (3.5, {0.5*sqrt(3)}) -- (4, {sqrt(3)})  ;
\draw[blue, dashed, very thick] (4.5, {0.5*sqrt(3)}) -- (5, {sqrt(3)})  ;
\draw[blue, dashed, very thick] (3, {sqrt(3)}) -- (4, {sqrt(3)})  ;
\draw[blue, dashed, very thick] (2.5, {0.5*sqrt(3)}) -- (5.5, {0.5*sqrt(3)})  ;
\draw[blue, dashed, very thick](4, {sqrt(3)}) -- (5, {sqrt(3)})  ;
\draw[blue, dashed, very thick] (5, {sqrt(3)}) -- (5.5, {0.5*sqrt(3)});
\draw[blue, dashed, very thick] (4, {sqrt(3)}) -- (4.5, {0.5*sqrt(3)});
\draw[blue, dashed, very thick] (3, {sqrt(3)}) -- (3.5, {0.5*sqrt(3)});
\draw[blue, dashed, very thick] (5.5, {0.5*sqrt(3)}) -- (5, 0)  ;
\draw[blue, dashed, very thick] (4.5, {0.5*sqrt(3)}) -- (4, 0)  ;
\draw[blue, dashed, very thick] (3.5, {0.5*sqrt(3)}) -- (3, 0)  ;
\draw[blue, dashed, very thick] (5, 0) -- (3, 0)  ;
\draw[blue, dashed, very thick] (3.5, {0.5*sqrt(3)}) -- (4, {sqrt(3)})  ;
\foreach \s in {0.5} 
{

\filldraw [black]  (3.5+\s, {0.5*sqrt(3)}) node   [mybox] (box) {\bf $x_0$}circle(0.1ex);
\filldraw [black]  (4.5+\s, {0.5*sqrt(3)}) node  {\colorbox{white}{$x_1$}} circle(0.1ex);
\filldraw [black]  (3+\s, {sqrt(3)}) node   [mybox] (box) {\bf $x_2$} circle(0.1ex);
\filldraw [black]  (3+\s, 0) node  {\colorbox{white}{$x_3$}} circle(0.1ex);
\filldraw [black]  (2.5+\s, {0.5*sqrt(3)}) node   [mybox] (box) {\bf $x_{23}$} circle(0.1ex);
\filldraw [black]  (4+\s, {sqrt(3)}) node  {\colorbox{white}{$x_{12}$}} circle(0.1ex);
\filldraw [black]  (4+\s, 0) node  {\colorbox{white}{$x_{13}$}} circle(0.1ex);
\filldraw [black]  (2.5+\s, {1.5*sqrt(3)}) node  {\colorbox{white}{$x_{22}$}} circle(0.1ex);
\filldraw [black]  (3.5+\s, {1.5*sqrt(3)}) node  {\colorbox{white}{$x_{122}$}} circle(0.1ex);
\filldraw [black]  (2.5, {1*sqrt(3)}) node  {\colorbox{white}{$x_{223}$}} circle(0.1ex);
}
\draw (3.5, {-0.6*sqrt(3)}) node [rotate=90] {$\curvearrowupdown$};
\draw (4, {-0.6*sqrt(3)}) node [rotate=90] {$\curvearrowupdown$};

\draw (0.8, {0.5*sqrt(3)}) node {$\curvearrowupdown$};
\draw (0.8, {sqrt(3)}) node {$\curvearrowupdown$};

\draw[blue, very thick] (3.65, {0.85*sqrt(3)}) node [rotate=25] { $\curvearrowleft$};
\draw[blue, very thick] (3.85, {0.65*sqrt(3)}) node [rotate=-155] {$\curvearrowleft$};
\draw (3.5, {-0.9*sqrt(3)}) node [rotate=90] {$w_1$};
\draw (4, {-0.9*sqrt(3)}) node [rotate=90] {$w_0$};

\draw (0.3, {0.5*sqrt(3)}) node {$s_2$};
\draw (0.3, {sqrt(3)}) node {$s_0$};
\draw (3.75, {0.75*sqrt(3)}) node 
{\colorbox{white}{$\pi$}} circle(0.1ex);
\draw[blue] (5, {sqrt(3)}) --  (2.5, {0.5*sqrt(3)});
\filldraw[blue!50, opacity=0.2] (3, {sqrt(3)}) -- (2.5, {0.5*sqrt(3)}) -- (4.5, {0.5*sqrt(3)})--(5, {sqrt(3)}) -- (3, {sqrt(3)});
\end{tikzpicture}

}
\subcaption{}\label{a1a1}
\end{subfigure}
\end{figure}

Define the function $x:  X_S \to \mb{C}$ on the vertices of the asymmetric cube,
we write
\begin{align}
 &x_S=x(X_S), \label{xfuna} 
\end{align}
and the elements $w\in W(B_2+A_1)$ act on the variables $x_S$ by
\begin{align}
 &w(x_S)=x(w.X_S)=x(X_{w.S})=x_{w.S}\label{wxfun1}.
\end{align}

As in the case of system \eqn{f} on the symmetric $3$-cube, the asymmetric system \eqn{af}  can be constructed from
four initial values. Here we have taken the initial value set to be $\{x_{0},x_{2},x_{12},x_{23}\}$ (indicated by the black nodes in Figure \subref{a3cube})
and use the equations in \eqn{af} to
express $x_1, x_3, x_{13}, x_{123}$ in terms of the initial values as follows.
For $x_1$ and $x_3$ use Equations \eqn{af12} and \eqn{af23}, respectively. For $x_{13}$, we use what is called the ``tetrahedron
equation'' (see the tetrahedron outlined by the densely dotted lines in Figure \subref{a3cube}),
\begin{align}\nonumber
 &T_{H}(x_0, x_{12}, x_{13}, x_{23}; \rho, \beta, \ga)  \\\label{TH}
 =&\beta ^2\ga ( x_0 x_{12}+  x_{13} x_{23})-\beta  \gamma ^2
   (x_0 x_{13}+x_{12} x_{23})-(\be^2-\gamma ^2) \rho x_0 x_{23}.
\end{align}

Equation \eqn{TH} is a consequence of the Equations in system \eqn{af}. It can be obtained by first getting the expressions for
$\frac{x_1}{x_3}$, $\frac{x_2}{x_1}$ and $\frac{x_3}{x_2}$ using Equations \eqn{af13}, \eqn{af12} and \eqn{af23},
respectively; then relation $\frac{x_1}{x_3}\frac{x_2}{x_1}\frac{x_3}{x_2}=1$ implies Equation \eqn{TH}.
Finally, for $x_{123}$ we use Equation \eqn{af132}.


The natural action of $W(B_2+A_1)$  on the index set $S$ of $X_S$ given in Equation \eqn{actac}
is extended to that on the variables $x_S$ by Equation \eqn{wxfun1}.
System \eqn {af} is reformulated as a birational representation
of $W(B_2+A_1)$ on the initial value set $\{x_{0},x_{2},x_{12},x_{23}\}$ by the following proposition.

{\prop \label{B2A1eqn} 
The actions on the initial value set $\{x_{0},x_{1},x_{12},x_{23}\}$ to be associated with the generators
of $W(B_2+A_1)$ are defined as follows,
\bee\nonumber
&&s_2:\{ x_2 \to \frac{x_2(x_0\ga-\be x_{23})}{x_0\be-\ga x_{23}},\\\nonumber
&&\hs{5mm} x_{12}\to \frac{-x_0\be^2\ga x_{12}+\be\ga^2x_{12}x_{23}-x_0x_{23}(\be^2-\ga^2)\rho}{\be\ga(-x_0\ga+\be x_{23})},
\be\leftrightarrow\ga\},\\\label{acta2a1}
&&s_3:\{x_0 \to \frac{x_2(x_0\ga-\be x_{23})}{x_0\be-\ga x_{23}}, \quad x_2 \leftrightarrow x_{23},\quad
x_{12}\to -\frac{x_2 x_{12}+x_2 x_{23}\rho/\ga}{x_{23}}\},\\\nonumber
&&w_0:\{x_0 \to -\frac{-\be x_0}{x_2(\be x_{12}+\rho x_0)}, \quad x_2 \to 1/x_{12},\quad 
 x_{23} \to -\frac{x_{23}}{x_2x_{12}+x_2x_{23}\rho/\ga}, \\\nonumber
&&\hspace{2cm} x_{12}\to1/x_2\}.
\ee
The transformations \eqn{acta2a1} satisfy the fundamental relations in \eqn{actB2A1} for $W(B_2+A_1)$.}

\begin{proof}
It can be checked by direct computation that the transformations satisfy the fundamental relations \eqn{actB2A1} of $W(B_2+A_1)$. 
System \eqn{af} of quad-equations can be generated from the two basic equations \eqn{a23}--\eqn{a12}
by the actions of the group elements of $W(B_2+A_1)$.
Equations \eqn{af23} and \eqn{af231}  can be
obtained from Equation \eqn{a23} by applying the group elements $\{1, w_0\}$, respectively; Equations \eqn{af12},
\eqn{af123}, \eqn{af13} and \eqn{af132} can be
obtained from Equation \eqn{a12} by applying the group elements $1, s_3, s_2$ and $s_2s_3$, respectively.
\end{proof}
Now we show that system \eqn{af} is related by reduction to a discrete \Pa equations with $\widetilde{W}\bigl((A_1+A'_1)^{(1)}\bigr)$ symmetry in Sakai's classification.
\subsection{Reduction to a
 $q$-discrete \Pa equation with $\widetilde{W}\bigl((A_1+A'_1)^{(1)}\bigr)$ symmetry}\label{DISA1A1}

On system \eqn{af} letting 
\beqn\label{agr}
x_{123}=x_0
\eeqn impose the following conditions on the variables and parameters,
\beqn\label{conpa}
t_{1}t_{2}t_{3}(x_S)=x_{S\cup \{1, 2, 3\}}=x_{S},
\quad \rho_1=q\rho, \quad \be_2=q\be,\quad \ga_3=q\ga,\quad q\in\mathbb{C}^\ast,
\eeqn
where $t_{j}(x_S)=x_{S\cup \{j\}}$ denotes the shift in direction $j$ on the $\widetilde{W}\bigl((A_1+A'_1)^{(1)}\bigr)$ lattice.

In section \ref{DISA2}, reduction condition \eqn{pc} on system \eqn{f}  of six same quad-equations corresponds to the orthogonal projection of a symmetric $3$-cube with $W(B_3)$ symmetry to a triangular lattice with $\widetilde{W}\bigl(A_2^{(1)}\bigr)$ symmetry (see Figures \subref{3cube} and \subref{cth}). Here for the asymmetric system \eqn {af}, reduction condition \eqn{agr} 
corresponds to the projection (indicated by the
loosely dotted line in Figures \subref{a3cube}) of an asymmetric $3$-cube with $W(B_2+A_1)$ symmetry to a rectangular lattice 
(see Figure \subref{a1a1}) with 
$\widetilde{W}\bigl((A_1+A'_1)^{(1)}\bigr)$ symmetry\footnote{The $'$ on $A'_1$ denotes the fact that these are two root systems of type $A_1$ with different root (weight) lengths. }.
 The extended affine Weyl group $\widetilde{W}\bigl((A_1+A'_1)^{(1)}\bigr)$ is generated by the two
sets of generators $\{s_2, s_0\}$ and $\{w_1, w_0\}$ that commute (realised as
two sets of orthogonal affine reflection lines); and $\pi$ the Dynkin diagram automorphism (realised as
rotation by $180^{\circ}$ about the main diagonal of the two adjacent quadrilaterals shaded in blue).
In other words, the generators satisfy
the following relations:
\bee\nonumber
&&s_0^2=s_2^2=(s_0s_2)^\infty=1, \quad w_0^2=w_1^2=(w_0w_2)^\infty=1,\\\label{acta1a1f}
&&(s_iw_j)^2=1, \quad i=0, 2,\; j=0, 1, \\\nonumber
&&\pi^2=1, \quad\pi w_0=w_1\pi,\quad \pi s_2=s_0\pi.
\ee
We note that the relation $(w w')^\infty=1$ means that
there is no positive integer $N$ such that $(w w')^N=1$
for transformations $w$ and $w'$.
The corresponding Dynkin diagram is given in Figure \subref{DA1A1}. The natural action of $w\in \widetilde{W}\bigl((A_1+A'_1)^{(1)}\bigr)$ on $X_S$, points of the rectangular $\widetilde{W}\bigl((A_1+A'_1)^{(1)}\bigr)$ lattice,  is extended to that on the variables $x_S$ as defined in Equation \eqn{wxfun1}. 

The first thing to notice is that the number of initial values of the reduced system becomes three.
We choose $\{x_0, x_2, x_{23}\}$, and express $x_{12}$ in terms of these using 
Equation \eqn{af132} with the reduction condition $x_{123}=x_0$. In fact it is easy to see that any $x$ variable on the lattice in Figure \subref{a1a1}
can be reach from  $\{x_0, x_2, x_{23}\}$ (indicated by the rounded boxes) by using the reflections $s_2, s_0, w_1, w_0$
and rotation $\pi$.
System \eqn{af} with constraint \eqn{conpa} can then be rewritten as
a birational representation of $\widetilde{W}\bigl((A_1+A'_1)^{(1)}\bigr)=\langle s_2, s_0, w_1, w_0, \pi\rangle$
on the initial value set $\{x_0, x_2, x_{23}\}$.
{\prop
 The actions to be associated with the generators $s_2, s_0, w_1, w_0$ on the initial value set $\{x_0, x_2, x_{23}\}$ are given by,
\bee\nonumber
&&\hs{-5mm}s_2:\{ x_2 \to \frac{x_2(x_0\ga-\be x_{23})}{x_0\be-\ga x_{23}}\},\\\nonumber
&&\hs{-5mm}s_0:\{x_0 \to \frac{x_0(\ga^2x_2^2+q\be\ga x_0x_{23}+q\be\rho x_2x_{23})}
{\ga(q \be x_2^2+x_{23}(\ga x_0+\rho x_2)},\,\\\nonumber
&&\hs{38mm} {x_{23} \to \frac{q\be x_{23}(q \be x_2^2+x_{23}(\ga x_0+\rho x_2))}{q\be\ga x_2^2+q^2\be^2 x_0x_{23}+\ga\rho x_2x_{23}}\},}\\\label{acta1a1v}
&&\\\nonumber
&&\hs{-5mm}w_0: \{x_0\to 1/x_0, \quad x_2 \to -\frac{\ga x_2}{x_{23}(\ga x_0+\rho x_2)},\\\nonumber
&&\hs{40mm}x_{23}\to \frac{\be\ga x_0}{\be \ga x_0x_{23}-\ga\rho x_0x_2+\be\rho x_2x_{23}}\},\\\nonumber
&&\hs{-5mm}w_1:\{x_0 \leftrightarrow 1/x_{23}, \quad  x_2 \to1/x_2\},\\\nonumber
&&\hs{-5mm}\pi:\{x_0\to 1/x_2, \quad x_2\to 1/x_0, \quad x_{23}\to -\frac{\ga x_2}{x_{23}(\ga x_0+\rho x_2)}\},
\ee
and the actions on the parameters are given by
\bee\nonumber
&&s_2:\{\be, \ga\}\mapsto \{\ga, \be\},\\\nonumber
&&s_0:\{\be, \ga\}\mapsto \{\ga/q, q \be\},\\\label{acta1a1p}
&&w_0:\{\be, \ga, q\}\mapsto \{\ga, \be, 1/q\},\\\nonumber
&&w_1:\{\be,\ga,\rho, q\}\mapsto \{\ga, \be,\rho/q, 1/q\},\\\nonumber
&&\pi:\{\be, q\}\mapsto \{q\be, 1/q\}.
\ee
The transformations given in Equations \eqn{acta1a1v}--\eqn{acta1a1p} satisfy the fundamental relations of 
$\widetilde{W}\bigl((A_1+A'_1)^{(1)}\bigr)$ in Equation \eqn{acta1a1f}. 
In particular, the translational elements $t_i$ ($i=1, 2, 3$) are defined by
\beqn\label{trana1a1}
t_1=w_0w_1, \quad t_2=\pi s_1w_0, \quad t_3=\pi s_0 w_0.
\eeqn
We note that $t_1t_2t_3:\{\rho\to q \rho,\,\be\to q \be, \,
\ga\to q \ga\}$. However, since the parameters in system \eqn{acta1a1v}
can always be arranged as pairs of ratios so that 
the condition on the translational elements, $t_1t_2t_3=1$, is satisfied. 
This in fact corresponds to our reduction condition \eqn{agr}.

\begin{figure}
\begin{subfigure}[b]{.3\textwidth}

\centering
 \resizebox{\linewidth}{!}{
\begin{tikzpicture}[scale=1]
\node  (a1) {$\circ$};
\node [right=of a1](a2) {$\circ$} ;
\node [right=of a2](a3) {$\circ$};
\draw (a1) node [anchor=north] {$w_0$} ;
\draw (a2) node [anchor=north] {$s_2$} ;
\draw (a3) node [anchor=north] {$s_3$} ;
\draw[double] (a2) -- node {} (a3);
\path[use as bounding box] (-1.5,0) rectangle (0,0);
\end{tikzpicture}
}
\subcaption{$\Ga(B_2+A_1)$}\label{DB2A1}
\end{subfigure}
\begin{subfigure}[b]{.5\textwidth}
\centering
 \resizebox{\linewidth}{!}{
\begin{tikzpicture}[scale=1]
\node  (a1) {$\circ$};
\node [right=of a1](a2) {} ;
\node [right=of a2](a3) {$\circ$} ;
\draw (a1) node [anchor=north] {$s_0$} ;
\draw (a2) node [anchor=north] {$\infty$} ;

\draw (a3) node [anchor=north] {$s_2$} ;
\draw[-] (a1) -- node {} (a3);
\node  (a4) [right=of a3]{$\circ$};
\node [right=of a4](a5) {} ;
\node [right=of a5](a6) {$\circ$} ;
\draw (a4) node [anchor=north] {$w_1$} ;
\draw (a5) node [anchor=north] {$\infty$} ;
\draw (a6) node [anchor=north] {$w_0$} ;
\node (r1) at ($(a1)!0.5!(a3)$) {};
\draw (r1) node [anchor=south] {$\pi$} ;
\node (r2) at ($(a4)!0.5!(a6)$) {};
\draw (r2) node [anchor=south] {$\pi$} ;
\draw (5.45, 0.5) node  [rotate=-90] {$\curvearrowupdown$};
\draw (1.35, 0.5) node  [rotate=-90] {$\curvearrowupdown$};
\draw[-] (a4) -- node {} (a6);
\path[use as bounding box] (-1.5,0) rectangle (0,0);
\end{tikzpicture}
}
\subcaption{$\widetilde{\Ga}^{(1)}(A_1+A_1')$}\label{DA1A1}
\end{subfigure}
\end{figure}
}
\begin{proof}
The relations between the $x$ variables on any quadrilateral of the $\widetilde{W}\bigl((A_1+A'_1)^{(1)}\bigr)$
 lattice in Figure \subref{a1a1} can
be expressed in terms of Equations \eqn{af23},  \eqn{af12}, \eqn{af132}, and their shifted versions under
 the constraint \eqn{conpa}. The
action of generator $s_0$ is $\{x_0 \to x_{122},\, x_{23} \to x_{22}\}$. We can express:
$x_{122}$ in terms of the initial values by shifting Equation \eqn{af23} by $t_{1}t_{2}$,
$x_{223}$ by shifting Equation \eqn{af12} by $t_{2}t_{3}$,
and finally $x_{22}$ by shifting Equation \eqn{af132} by $t_{2}$.
The action of generator $w_0$ is: $\{x_0 \to 1/x_0, x_2 \to 1/x_{12}, x_{23} \to 1/x_{1}\}$, which again can be
expressed in terms of the initial values.
\end{proof}
The dynamics of system defined by transformations \eqn{acta1a1v}--\eqn{acta1a1p} are described by the translational elements
of $\widetilde{W}\bigl((A_1+A'_1)^{(1)}\bigr)$. We show in the following proposition that it is equivalent to a $q$-discrete \Pa equation.


{\prop
Let
\beqn\label{QPA11}
f=\frac{x_0}{x_2},\quad g=\frac{x_2}{x_{23}},
\eeqn
then the reduced system \eqn{af} with constraint \eqn{conpa} can be rewritten as
\beqn\label{QPA11e}
f_1=-\frac{1}{f}(1+\frac{\rho}{\be g_1}),\quad g_1=-\frac{1}{g}(1+\frac{\rho}{\ga f}),
\eeqn
where $t_1=w_0w_1$ describes right horizontal translation on the $\widetilde{W}\bigl((A_1+A'_1)^{(1)}\bigr)$ lattice, and 
$t_1(\rho)=\rho_1=q\rho$, $t_1(f)=\ol f$, and $t_1(g)=\ol g$.
System \eqn{QPA11e} is a second-order nonlinear ordinary $q$-discrete
equation of \Pa type with  $\widetilde{W}\bigl((A_1+A'_1)^{(1)}\bigr)$ symmetry, first obtained in \cite{KTGR:00}.
}
\begin{proof}
System \eqn{QPA11e} can be obtained by applying the action of $t_1=w_0w_1$ on $f$ and $g$ using  
the transformations  given in Equations \eqn{acta1a1v}--\eqn{acta1a1p}. Alternatively they can be obtained
by rewriting Equation (\ref{af12}) and Equation (\ref{af132}) with constraint \eqn{conpa}.
\end{proof}
{\Rem
System \eqn{af} is related to the discrete \Pa equation with
$\widetilde{W}\bigl((A_1+A'_1)^{(1)}\bigr)$ symmetry in Lemma 4.2 of \cite{jns1} by a gauge transformation.
}
\section{An algebro-geometric approach}\label{ag}
In this section, we explain in detail 
how to construct a representation of the extended affine Weyl group $W\bigl(A^{(1)}\bigr)$ associated with the Cartan
matrix $A^{(1)}$ on the field of rational functions
and thereby derive discrete \Pa equations as representations of $\widetilde{W}(A^{(1)})$
on the spaces of point configurations in projective space. In particular, we shown how discrete \Pa equations arise
as Cremona transformations from the  module \eqn{eqn:module10}.

In Section \ref{DIS}, we interpreted the dynamics of discrete integrable systems as translations on the weight
lattice by the actions of $W\bigl(A^{(1)}\bigr)$ on
a hyperplane $H\subset V^{(1)\ast}$. In this section, we look at the linear actions (in the sense explained
in Section \ref{twopic})
on the set of simple affine roots $\{ \al_0, ..., \al_n\}$ in $V^{(1)}$. 
To emphasise the fact that transformations are associated with the
linear actions of $\widetilde{W}(A^{(1)})$  on $V^{(1)}$ we use the following convention
\beqn\label{convd}
f.w=w^{-1}.f, \quad \mbox{for}\quad f\in V^{(1)}.
\eeqn
That is, we use the right action to indicate the fact that we are acting on $V^{(1)}$, whereas
left actions are associated with transformations on the affine hyperplane $H\subset  V^{(1)\ast}$.

The action of the translational elements on $V^{(1)}$, which we named ``shifted motion'',
is given by Equations \eqn{Transl}, satisfying
condition \eqn{Tranc}. This shifted motion on the simple affine roots is associated with the transformations of 
parameters of the discrete Painlev\'e equations \cite{sak:01}. For more background on the formulation of the discrete \Pa equations as
birational representations of the Weyl groups, in particular on the
discussions of the translational elements of the affine Weyl groups associated
with the discrete \Pa equations see \cite{NY:98, KMNOY:06, sak:01}.
Here we give the a $(A_2+A_1)^{(1)}$-type discrete \Pa equation as an example. 

Let us consider the following module over $\mathbb{Z}$:
\begin{equation}\label{eqn:module10}
 \bigoplus_{i=1}^2\mathbb{Z}H_i\bigoplus_{i=1}^8\mathbb{Z}E_i,
\end{equation}
where $\{H_1,H_2,E_1,\dots,E_8\}$ is a basis of the module, equipped with the symmetric bilinear form $(\,|\,)$ given by
\begin{equation}
 (H_i|H_j)=1-\delta_{ij},\quad
 (H_i|E_j)=0,\quad
 (E_i|E_j)=-\delta_{ij}.
\end{equation}
We define the root system of type $(A_2+A_1)^{(1)}$ by
\begin{equation}\label{eqn:root_A2A1}
 Q((A_2+A_1)^{(1)})
 =\mathbb{Z}\alpha_0\bigoplus\mathbb{Z}\alpha_1\bigoplus\mathbb{Z}\alpha_2
 \bigoplus\mathbb{Z}\beta_0\bigoplus\mathbb{Z}\beta_1,
\end{equation}
where
\begin{subequations}
\begin{align}
 &\alpha_0=H_1-E_1-E_4,\quad
 \alpha_1=H_2-E_2-E_5,\quad
 \alpha_2=H_1+H_2-E_3-E_6-E_7-E_8,\\
 &\beta_0=H_1+H_2-E_2-E_4-E_6-E_8,\quad
 \beta_1=H_1+H_2-E_1-E_3-E_5-E_7.
\end{align}
\end{subequations}
Note that the corresponding Cartan matrices of $\bigoplus_{i=0}^2\mathbb{Z}\alpha_i$ and $\bigoplus_{i=0}^1\mathbb{Z}\beta_i$
are of the $A_2^{(1)}$- and $A_1^{(1)}$-types, respectively:
\begin{equation}
 (a_{ij})_{i,j=0}^2=\begin{pmatrix}2&-1&-1\\-1&2&-1\\-1&-1&2\end{pmatrix},\quad
 (b_{ij})_{i,j=0}^1=\begin{pmatrix}2&-2\\-2&2\end{pmatrix},
\end{equation}
where
\begin{equation}
 a_{ij}=\cfrac{2(\alpha_i|\alpha_j)}{(\alpha_j|\alpha_j)},\quad
 b_{ij}=\cfrac{2(\beta_i|\beta_j)}{(\beta_j|\beta_j)}.
\end{equation}

\subsection{Picard group}
We give the algebro-geometric meaning to the module \eqref{eqn:module10}.
Let $(f,g)$ be inhomogeneous coordinates of $\mathbb{P}^1\times\mathbb{P}^1$
and $P_i\in\mathbb{P}^1\times\mathbb{P}^1$ be the following points:
\begin{subequations}\label{eqns:A5_bps}
\begin{align}
 &P_1: (f,g)=(-{a_0}^{-1},0),\\
 &P_2: (f,g)=(0,-a_1),\\
 &P_3: (f,g)=(0,\infty),
 &&P_7: (f,g;fg)=(0,\infty;-c^2a_0a_1{a_2}^2),\\
 &P_4: (f,g)=(-a_0,\infty),\\
 &P_5: (f,g)=(\infty,-{a_1}^{-1}),\\
 &P_6:(f,g)=(\infty,0),
 &&P_8: (f,g;fg)=(\infty,0;-c^2a_0a_1),
\end{align}
\end{subequations}
where $a_0,a_1,a_2,c\in\mathbb{C}^\ast$.
Let $X$ be the rational surface obtained by the blowing-up at the eight points, $\epsilon: X \to \mathbb{P}^1\times\mathbb{P}^1$,
and $E_i=\epsilon^{-1}(P_i)$, $i=1,\dots,8$, and $H_j$, $j=1,2$, be  the linear equivalence classes of
the total transform of the point of the $i$-th blow up and the coordinate lines $f$=constant and $g$=constant, respectively. 
Then, the module \eqref{eqn:module10} is called the Picard group of the rational surface $X$ and denoted by ${\rm Pic}(X)$:
\begin{equation}
 {\rm Pic}(X)=\bigoplus_{i=1}^2\mathbb{Z}H_i\bigoplus_{i=1}^8\mathbb{Z}E_i.
\end{equation}
Moreover, the bilinear form $(\,|\,)$ corresponds to the intersection form.

The anti-canonical divisor of $X$, denoted by $-K_X$, corresponds to the null root $\de$ 
of the affine root system defined by Equation \eqn{de} in Section \ref{twopic}. It is
uniquely decomposed into the prime divisors (or the simple affine roots of Section \ref{twopic}):
\begin{equation}\label{eqn:A5_-KX_decomposed}
\de= -K_X=2H_1+2H_2-\sum_{i=1}^8E_i=\sum_{i=0}^5D_i,
\end{equation}
where
\begin{subequations}\label{eqns:A5_Di}
\begin{align}
 &D_0=E_6-E_8,\quad
 D_1=H_2-E_1-E_6,\quad
 D_2=H_1-E_2-E_3,\\
 &D_3=E_3-E_7,\quad
 D_4=H_2-E_3-E_4,\quad
 D_5=H_1-E_5-E_6.
\end{align}
\end{subequations}
The submodule of ${\rm Pic}(X)$
\begin{equation}\label{eqn:root_A5}
 Q(A_5^{(1)})=\bigoplus_{i=0}^5\mathbb{Z}D_i,
\end{equation}
is the root system of type $A_5^{(1)}$ since its corresponding Cartan matrix is of type $A_5^{(1)}$:
\begin{equation}
 (d_{ij})_{i,j=0}^5
 =\begin{pmatrix}
 2&-1&0&0&0&-1\\
 -1&2&-1&0&0&0\\
 0&-1&2&-1&0&0\\
 0&0&-1&2&-1&0\\
 0&0&0&-1&2&-1\\
 -1&0&0&0&-1&2
 \end{pmatrix},
\end{equation}
where
\begin{equation}
 d_{ij}=\cfrac{2(D_i|D_j)}{(D_j|D_j)}.
\end{equation}
The corresponding Dynkin diagram is given in Figure \ref{N3} (where we have used the simple roots instead of the simple reflections
to denote the nodes).
Note that $mH_1+nH_2-\sum_{i=1}^8\mu_iE_i$ corresponds to a curve of bi-degree $(m,n)$ on $\mathbb{P}^1\times\mathbb{P}^1$
($(m,n)$-curve of $(f,g)$)
passing through the base points $P_i$ with multiplicity $\mu_i$.
Moreover, $E_i-E_j$ corresponds to an exceptional curve inserted into the base point $P_i$, 
which passes through the base point $P_j$.
Therefore, Equation \eqref{eqn:A5_-KX_decomposed} means that 
the $(2,2)$-curve passing through the eight base points $P_i$ is decomposed into 
the six curves on the surface $X$. 
This decomposition, thus the position of the base points, characterises the type of rational surface $X$.
Therefore, we refer to the surface $X$ as $A_5^{(1)}$-surface.
The most general case is $A_0^{(1)}$-surface and its degenerations are illustrated in Figure \ref{N1}, first
given by Sakai \cite{sak:01}.
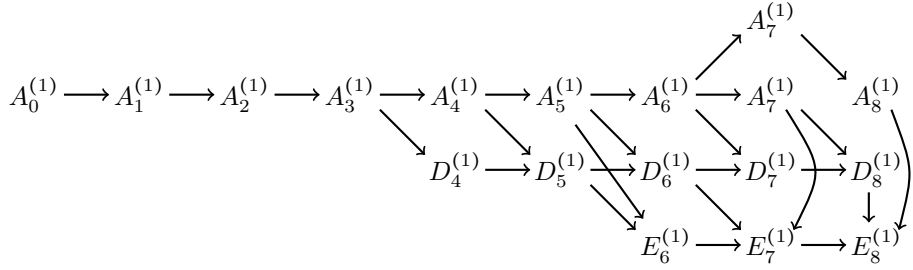
\begin{figure}[t]
\begin{center}
\begin{tikzpicture}[scale = 1]
\begin{scope}
\coordinate (P11s) at (0,0);
\coordinate (P11e) at ($(P11s)+(0.6,0)$);
\coordinate (P12s) at ($(P11e)+(0.8,0)$);
\coordinate (P12e) at ($(P12s)+(0.6,0)$);
\coordinate (P13s) at ($(P12e)+(0.8,0)$);
\coordinate (P13e) at ($(P13s)+(0.6,0)$);
\coordinate (P14s) at ($(P13e)+(0.8,0)$);
\coordinate (P14e) at ($(P14s)+(0.6,0)$);
\coordinate (P15s) at ($(P14e)+(0.8,0)$);
\coordinate (P15e) at ($(P15s)+(0.6,0)$);
\coordinate (P16s) at ($(P15e)+(0.8,0)$);
\coordinate (P16e) at ($(P16s)+(0.6,0)$);
\coordinate (P17s) at ($(P16e)+(0.8,0)$);
\coordinate (P17e) at ($(P17s)+(0.6,0)$);
\coordinate (P18s) at ($(P17e)+(0.8,0)$);
\coordinate (P18e) at ($(P18s)+(0.6,0)$);
\coordinate (P19s) at ($(P18e)+(0.8,0)$);
\coordinate (P19e) at ($(P19s)+(0.6,0)$);
\coordinate (P21s) at (0,-1);
\coordinate (P21e) at ($(P21s)+(0.6,0)$);
\coordinate (P22s) at ($(P21e)+(0.8,0)$);
\coordinate (P22e) at ($(P22s)+(0.6,0)$);
\coordinate (P23s) at ($(P22e)+(0.8,0)$);
\coordinate (P23e) at ($(P23s)+(0.6,0)$);
\coordinate (P24s) at ($(P23e)+(0.8,0)$);
\coordinate (P24e) at ($(P24s)+(0.6,0)$);
\coordinate (P25s) at ($(P24e)+(0.8,0)$);
\coordinate (P25e) at ($(P25s)+(0.6,0)$);
\coordinate (P26s) at ($(P25e)+(0.8,0)$);
\coordinate (P26e) at ($(P26s)+(0.6,0)$);
\coordinate (P27s) at ($(P26e)+(0.8,0)$);
\coordinate (P27e) at ($(P27s)+(0.6,0)$);
\coordinate (P28s) at ($(P27e)+(0.8,0)$);
\coordinate (P28e) at ($(P28s)+(0.6,0)$);
\coordinate (P29s) at ($(P28e)+(0.8,0)$);
\coordinate (P29e) at ($(P29s)+(0.6,0)$);
\coordinate (P31s) at (0,-2);
\coordinate (P31e) at ($(P31s)+(0.6,0)$);
\coordinate (P32s) at ($(P31e)+(0.8,0)$);
\coordinate (P32e) at ($(P32s)+(0.6,0)$);
\coordinate (P33s) at ($(P32e)+(0.8,0)$);
\coordinate (P33e) at ($(P33s)+(0.6,0)$);
\coordinate (P34s) at ($(P33e)+(0.8,0)$);
\coordinate (P34e) at ($(P34s)+(0.6,0)$);
\coordinate (P35s) at ($(P34e)+(0.8,0)$);
\coordinate (P35e) at ($(P35s)+(0.6,0)$);
\coordinate (P36s) at ($(P35e)+(0.8,0)$);
\coordinate (P36e) at ($(P36s)+(0.6,0)$);
\coordinate (P37s) at ($(P36e)+(0.8,0)$);
\coordinate (P37e) at ($(P37s)+(0.6,0)$);
\coordinate (P38s) at ($(P37e)+(0.8,0)$);
\coordinate (P38e) at ($(P38s)+(0.6,0)$);
\coordinate (P39s) at ($(P38e)+(0.8,0)$);
\coordinate (P39e) at ($(P39s)+(0.6,0)$);
\coordinate (P41s) at (0,-3);
\coordinate (P41e) at ($(P41s)+(0.6,0)$);
\coordinate (P42s) at ($(P41e)+(0.8,0)$);
\coordinate (P42e) at ($(P42s)+(0.6,0)$);
\coordinate (P43s) at ($(P42e)+(0.8,0)$);
\coordinate (P43e) at ($(P43s)+(0.6,0)$);
\coordinate (P44s) at ($(P43e)+(0.8,0)$);
\coordinate (P44e) at ($(P44s)+(0.6,0)$);
\coordinate (P45s) at ($(P44e)+(0.8,0)$);
\coordinate (P45e) at ($(P45s)+(0.6,0)$);
\coordinate (P46s) at ($(P45e)+(0.8,0)$);
\coordinate (P46e) at ($(P46s)+(0.6,0)$);
\coordinate (P47s) at ($(P46e)+(0.8,0)$);
\coordinate (P47e) at ($(P47s)+(0.6,0)$);
\coordinate (P48s) at ($(P47e)+(0.8,0)$);
\coordinate (P48e) at ($(P48s)+(0.6,0)$);
\coordinate (P49s) at ($(P48e)+(0.8,0)$);
\coordinate (P49e) at ($(P49s)+(0.6,0)$);
\node at ($(P18s)-(0.4,0)$){$A_7^{(1)}$};
\node at ($(P21s)-(0.4,0)$){$A_0^{(1)}$};
\node at ($(P22s)-(0.4,0)$){$A_1^{(1)}$};
\node at ($(P23s)-(0.4,0)$){$A_2^{(1)}$};
\node at ($(P24s)-(0.4,0)$){$A_3^{(1)}$};
\node at ($(P25s)-(0.4,0)$){$A_4^{(1)}$};
\node at ($(P26s)-(0.4,0)$){$A_5^{(1)}$};
\node at ($(P27s)-(0.4,0)$){$A_6^{(1)}$};
\node at ($(P28s)-(0.4,0)$){$A_7^{(1)}$};
\node at ($(P29s)-(0.4,0)$){$A_8^{(1)}$};
\node at ($(P35s)-(0.4,0)$){$D_4^{(1)}$};
\node at ($(P36s)-(0.4,0)$){$D_5^{(1)}$};
\node at ($(P37s)-(0.4,0)$){$D_6^{(1)}$};
\node at ($(P38s)-(0.4,0)$){$D_7^{(1)}$};
\node at ($(P39s)-(0.4,0)$){$D_8^{(1)}$};
\node at ($(P47s)-(0.4,0)$){$E_6^{(1)}$};
\node at ($(P48s)-(0.4,0)$){$E_7^{(1)}$};
\node at ($(P49s)-(0.4,0)$){$E_8^{(1)}$};
\draw [->, thick] (P21s)--(P21e);
\draw [->, thick] (P22s)--(P22e);
\draw [->, thick] (P23s)--(P23e);
\draw [->, thick] (P24s)--(P24e);
\draw [->, thick] (P25s)--(P25e);
\draw [->, thick] (P26s)--(P26e);
\draw [->, thick] (P27s)--(P27e);
\draw [->, thick] (P35s)--(P35e);
\draw [->, thick] (P36s)--(P36e);
\draw [->, thick] (P37s)--(P37e);
\draw [->, thick] (P38s)--(P38e);
\draw [->, thick] (P47s)--(P47e);
\draw [->, thick] (P48s)--(P48e);
\draw [->, thick] ($(P27s)+(0,0.2)$)--($(P17e)-(0,0.2)$);
\draw [->, thick] ($(P18s)-(0,0.2)$)--($(P28e)+(0,0.2)$);
\draw [->, thick] ($(P24s)-(0,0.2)$)--($(P34e)+(0,0.2)$);
\draw [->, thick] ($(P25s)-(0,0.2)$)--($(P35e)+(0,0.2)$);
\draw [->, thick] ($(P26s)-(0,0.2)$)--($(P36e)+(0,0.2)$);
\draw [->, thick] ($(P27s)-(0,0.2)$)--($(P37e)+(0,0.2)$);
\draw [->, thick] ($(P28s)-(0,0.2)$)--($(P38e)+(0,0.2)$);
\draw [->, thick] ($(P36s)-(0,0.2)$)--($(P46e)+(0,0.2)$);
\draw [->, thick] ($(P37s)-(0,0.2)$)--($(P47e)+(0,0.2)$);
\draw [->, thick] ($(P38e)+(0.3,-0.3)$)--($(P48e)+(0.3,0.3)$);
\draw [->, thick] ($(P26s)-(0.2,0.4)$)--($(P46e)+(0.1,0.35)$);
\draw[->, thick] ($(P28s)-(0.2,0.2)$) .. controls ($(P38s)+(0.3,0)$) .. ($(P48s)+(-0.1,0.2)$);
\draw[->, thick] ($(P29s)-(0.2,0.2)$) .. controls ($(P39s)+(0.1,0)$) .. ($(P49s)+(-0.1,0.2)$);

\end{scope}
\end{tikzpicture}
\caption{Degeneration of types of rational surfaces}\label{N1}
\end{center}
\end{figure}


The root systems orthogonal to those associated with surfaces are also important.
In the case of  $Q(A_5^{(1)})$, defined in Equation \eqref{eqn:root_A5}, the orthogonal root system is $Q((A_2+A_1)^{(1)})$ 
defined in Equation \eqref{eqn:root_A2A1}.
The cascade of the orthogonal root system corresponding to the degenerations of surfaces in Figure \ref{N1}
was given earlier in Figure \ref{N2}.

\subsection{Cremona isometries}
We consider the Cremona isometries for the $A_5^{(1)}$-surface $X$.
A Cremona isometry is defined by an automorphism of Pic$(X)$ which preserves 
\begin{description}
\item[(i)]
the intersection form on Pic$(X)$;
\item[(ii)]
the canonical divisor $K_X$;
\item[(iii)]
effectiveness of each effective divisor of Pic$(X)$.
\end{description}
The reflections for simple roots $\alpha_i$, $i=0,1,2$, and $\beta_i$, $i=0,1$,
and automorphisms of the Dynkin diagram corresponding to the divisors $D_i$, $i=0,\dots,5$,
are Cremona isometries and form the extended affine Weyl group of type $(A_2+A_1)^{(1)}$ \cite{sak:01}.

We define the right actions of the reflections $s_i$, $i=0,1,2$, and $w_j$, $j=0,1$, respectively across the hyperplane orthogonal to the root $\alpha_i$, $i=0,1,2$, 
and $\beta_j$, $j=0,1$, by the following:
\begin{equation}\label{ean:def_si_Picard}
 v.s_i=v-\cfrac{2(v|\alpha_i)}{(\alpha_i|\alpha_i)}\,\alpha_i,\quad
 v.w_j=v-\cfrac{2(v|\beta_j)}{(\beta_j|\beta_j)}\,\beta_j
\end{equation}
for all $v\in {\rm Pic}(X)$.
We also define the right actions of the diagram automorphisms Aut$(A_5^{(1)})=\langle \sigma\rangle$:
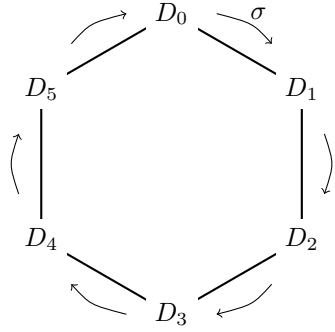
\begin{figure}[t]
\begin{center}
\begin{tikzpicture}[scale = 1]
\begin{scope}
\coordinate (P0) at (0,2);
\coordinate (P1) at (1.732,1);
\coordinate (P2) at (1.732,-1);
\coordinate (P3) at (0,-2);
\coordinate (P4) at (-1.732,-1);
\coordinate (P5) at (-1.732,1);
\coordinate (P01) at ($(1.732/1.6,3/1.6)$);
\coordinate (P12) at ($(2*1.732/1.6,0)$);
\coordinate (P23) at ($(1.732/1.6,-3/1.6)$);
\coordinate (P34) at ($(-1.732/1.6,-3/1.6)$);
\coordinate (P45) at ($(-2*1.732/1.6,0)$);
\coordinate (P50) at ($(-1.732/1.6,3/1.6)$);
\coordinate (PP01) at ($(1.732/1.5,3/1.5)$);
\coordinate (PP12) at ($(2*1.732/1.5,0)$);
\coordinate (PP23) at ($(1.732/1.5,-3/1.5)$);
\coordinate (PP34) at ($(-1.732/1.5,-3/1.5)$);
\coordinate (PP45) at ($(-2*1.732/1.5,0)$);
\coordinate (PP50) at ($(-1.732/1.5,3/1.5)$);
\draw [thick] (P0)--(P1);
\draw [thick] (P1)--(P2);
\draw [thick] (P2)--(P3);
\draw [thick] (P3)--(P4);
\draw [thick] (P4)--(P5);
\draw [thick] (P5)--(P0);
\node[fill=white] at (P0) {$D_0$};
\node[fill=white] at (P1) {$D_1$};
\node[fill=white] at (P2) {$D_2$};
\node[fill=white] at (P3) {$D_3$};
\node[fill=white] at (P4) {$D_4$};
\node[fill=white] at (P5) {$D_5$};
\node[fill=white] at (PP01) {$\sigma$};
\draw[->] ($(P0)+(0.6,0)$) .. controls (P01) .. ($(P1)+(-0.4,0.6)$);
\draw[->] ($(P1)+(0.3,-0.6)$) .. controls (P12) .. ($(P2)+(0.3,0.6)$);
\draw[->] ($(P2)+(-0.4,-0.6)$) .. controls (P23) .. ($(P3)+(0.6,0)$);
\draw[->] ($(P3)+(-0.6,0)$) .. controls (P34) .. ($(P4)+(0.4,-0.6)$);
\draw[->] ($(P4)+(-0.3,0.6)$) .. controls (P45) .. ($(P5)+(-0.3,-0.6)$);
\draw[->] ($(P5)+(0.4,0.6)$) .. controls (P50) .. ($(P0)+(-0.6,0)$);
\end{scope}
\end{tikzpicture}
\caption{Dynkin diagram of $Q(A_5^{(1)})$}\label{N3}
\end{center}
\end{figure}
by
\begin{align}
 &(H_1,H_2,E_1,E_2,E_3,E_4,E_5,E_6,E_7,E_8).\sigma\notag\\
 &=(H_2,H_1+H_2-E_3-E_6,E_2,E_7,H_2-E_3,E_5,E_8,H_2-E_6,E_4,E_1).
\end{align}
We can easily verify that the linear actions of 
$$\widetilde{W}((A_2+A_1)^{(1)})=\langle s_0,s_1,s_2,w_0,w_1,\sigma\rangle$$
on Pic$(X)$ satisfy the fundamental relations of the extended affine Weyl group of type $(A_2+A_1)^{(1)}$:
\begin{subequations}\label{eqn:fundamental_A2A1}
\begin{align}
 &{s_i}^2=(s_i s_{i+1})^3=1,\quad
 {w_j}^2=(w_j  w_{j+1})^\infty=1,\quad
 (s_i  w_j)^2=1,\\
 &\sigma^6=1,\quad
 s_i \sigma=\sigma  s_{i+1},\quad
 w_j \sigma=\sigma  w_{j+1},
\end{align}
\end{subequations}
where $i\in\mathbb{Z}/3\mathbb{Z}$ and $j\in\mathbb{Z}/2\mathbb{Z}$.

Define
\begin{equation}
 \pi=\sigma^2,\quad
 r=\sigma^3,
\end{equation}
the right actions of $\pi$ and $r$ on the simple roots of $Q((A_2+A_1)^{(1)})$ are given by
\begin{subequations}
\begin{align}
 &(\alpha_0,\alpha_1,\alpha_2,\beta_0,\beta_1).\pi=(\alpha_2,\alpha_0,\alpha_1,\beta_0,\beta_1),\\
 &(\alpha_0,\alpha_1,\alpha_2,\beta_0,\beta_1).r=(\alpha_0,\alpha_1,\alpha_2,\beta_1,\beta_0).
\end{align}
\end{subequations}

The corresponding Dynkin diagram is given in Figure \ref{a2a1}.
\begin{figure}
\begin{tikzpicture}
\node  (a1) {$\circ$};
\node [right=of a1](a2) {} ;
\node [right=of a2](a3) {$\circ$} ;
\node (a4) at (1.4, -0.9) {$\circ$} ;
\node [right=of a3](a5) {$\circ$} ;
\node [right=of a5](a6) {$\circ$} ;
\draw (a1) node [anchor=south] {$\al_1$} ;
\draw (a2) node [anchor=south] {} ;
\draw (a3) node [anchor=south] {$\al_{2}$} ;
\draw (a4) node [anchor=north] {$\al_{0}$} ;
\draw[-] (a1) -- node {} (a3);
\draw[-] (a1) -- node {} (a4);
\draw[-] (a3) -- node {} (a4);
\path[use as bounding box] (-1.5,0) rectangle (0,0);
\draw (1.45, 0.6) node  {$\curvearrowleft$};
\draw (0.4, -0.9) node [rotate=135] {$\curvearrowleft$};
\draw (2.4, -0.9) node [rotate=-135] {$\curvearrowleft$};
\draw  (1.4, 0.3) node {$\pi$};
\draw (a5) node [anchor=north] {$\be_1$} ;
\draw (a6) node [anchor=north] {$\be_0$} ;
\draw[-] (a5) -- node {} (a6);
\node (r) at ($(a5)!0.5!(a6)$) {};
\draw (r) node [anchor=south] {$r$} ;
\draw (r) node [anchor=north] {$\infty$} ;
\draw (4.8, 0.5) node  [rotate=-90] {$\curvearrowupdown$};
\end{tikzpicture}
\caption{Dynkin diagram for $Q((A_2+A_1)^{(1)})$ }\label{a2a1}
\end{figure}
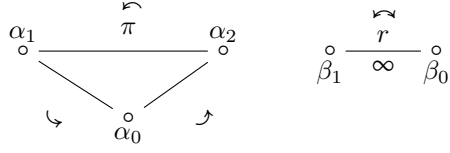

The action of $\widetilde{W}((A_2+A_1)^{(1)})$ is lifted to the level of the birational action
on the variables $f,g$ and the parameters $a_i$, $i=0,1,2$, and $c$ as stated in the following lemma.

\begin{lemma}\label{lemma:birational_action_A2A1}
Let
\begin{equation}
 q=a_0a_1a_2,\quad
 f_0=f,\quad f_1=g,\quad f_2=\cfrac{qc^2}{fg}.
\end{equation}
The left actions of $\widetilde{W}((A_2+A_1)^{(1)})$ on parameters are given by
\begin{align*}
 &s_i:(a_i,a_{i+1},a_{i+2},c)\mapsto({a_i}^{-1},a_ia_{i+1},a_ia_{i+2},c),
 &&\pi:(a_0,a_1,a_2,c)\mapsto(a_1,a_2,a_0,c),\\
 &w_0:(a_0,a_1,a_2,c)\mapsto(a_0,a_1,a_2,c^{-1}),
 &&w_1:(a_0,a_1,a_2,c)\mapsto(a_0,a_1,a_2,q^{-2}c^{-1}),\\
 &r:(a_0,a_1,a_2,c)\mapsto(a_0,a_1,a_2,q^{-1}c^{-1}),
\end{align*}
where $i\in\mathbb{Z}/3\mathbb{Z}$\,,
while its actions on variables are given by
\begin{align*}
 &s_i(f_{i-1})=\cfrac{f_{i-1}(1+a_if_i)}{a_i+f_i},\quad
 s_i(f_i)=f_i,\quad
 s_i(f_{i+1})=\cfrac{f_{i+1}(a_i+f_i)}{1+a_if_i},\quad
 \pi(f_i) = f_{i+1},\\
 &w_0(f_i)=\cfrac{a_ia_{i+1}(a_{i-1}a_i+a_{i-1}f_i+f_{i-1}f_i)}
  {f_{i-1}(a_ia_{i+1}+a_if_{i+1}+f_if_{i+1})},\\
 &w_1(f_i)=\cfrac{1+a_if_i+a_ia_{i+1}f_if_{i+1}}
  {a_ia_{i+1}f_{i+1}(1+a_{i-1}f_{i-1}+a_{i-1}a_if_{i-1}f_i)},\quad
  r(f_i)=\cfrac{1}{f_i},
\end{align*}
where $i\in\mathbb{Z}/3\mathbb{Z}$\,.
Note that for a function $F=F(a_0,a_1,a_2,c,f,g)$, we let an element $w\in\widetilde{W}((A_2+A_1)^{(1)})$ act as
\begin{equation}
 w.F=F(w.a_0,w.a_1,w.a_2,w.c,w.f,w.g),
\end{equation}
and the parameter $q$ is invariant under the action of $\widetilde{W}((A_2+A_1)^{(1)})$.
\end{lemma}
\begin{proof}
We consider the birational action of $s_0$
and denote its action by
\begin{equation}
 \overline{f}=s_0(f),\quad
 \overline{g}=s_0(g),\quad
 \overline{a}_i=s_0(a_i),~i=0,1,2,\quad
 \overline{c}=s_0(c).
\end{equation} 
We claim
\begin{equation}\label{eqn:claim_s0}
 \left(\begin{matrix}\overline{a}_0,\overline{a}_1,\overline{a}_2\\\overline{c}\end{matrix}\,;\overline{f},\overline{g}\right)
 =\left(\begin{matrix}{a_0}^{-1},a_0a_1,a_0a_2\\c\end{matrix}\,;f,\cfrac{g(a_0+f)}{1+a_0f}\right).
\end{equation}
We use the following notation.
Let $\overline{X}$ be the rational surface obtained by the blowing-up at the following base points:
\begin{subequations}
\begin{align}
 &\overline{P}_1: (\overline{f},\overline{g})=(-{\overline{a}_0}^{-1},0),\\
 &\overline{P}_2: (\overline{f},\overline{g})=(0,-\overline{a}_1),\\
 &\overline{P}_3: (\overline{f},\overline{g})=(0,\infty),
 &&\overline{P}_7: (\overline{f},\overline{g};\overline{f}\overline{g})=(0,\infty;-\overline{c}^2\overline{a}_0\overline{a}_1{\overline{a}_2}^2),\\
 &\overline{P}_4: (\overline{f},\overline{g})=(-\overline{a}_0,\infty),\\
 &\overline{P}_5: (\overline{f},\overline{g})=(\infty,-{\overline{a}_1}^{-1}),\\
 &\overline{P}_6:(\overline{f},\overline{g})=(\infty,0),
 &&\overline{P}_8: (\overline{f},\overline{g};\overline{f}\overline{g})=(\infty,0;-\overline{c}^2\overline{a}_0\overline{a}_1),
\end{align}
\end{subequations}
and $\overline{E}_i$, $i=1,\dots,8$, and $\overline{H}_j$, $j=1,2$, 
be the linear equivalence classes of the total transform of the point of the $i$-th blow up and 
the coordinate lines $\overline{f}$=constant and $\overline{g}$=constant, respectively. 
Definition \eqref{ean:def_si_Picard} gives the following actions:
\begin{subequations}
\begin{align}
 &H_1.s_0=\overline{H}_1,\quad
 H_2.s_0=\overline{H}_1+\overline{H}_2-\overline{E}_1-\overline{E}_4,\quad
 E_i.s_0=\overline{E}_i,~i\neq 1,4,\\
 &E_1.s_0=\overline{H}_1-\overline{E}_4,\quad
 E_4.s_0=\overline{H}_1-\overline{E}_1,
\end{align}
\end{subequations}
and
\begin{subequations}
\begin{align}
 &\overline{H}_1.s_0=H_1,\quad
 \overline{H}_2.s_0=H_1+H_2-E_1-E_4,\quad
 \overline{E}_i.s_0=E_i,~i\neq 1,4,\\
 &\overline{E}_1.s_0=H_1-E_4,\quad
 \overline{E}_4.s_0=H_1-E_1.
\end{align}
\end{subequations}
Since the relations 
\begin{subequations}
\begin{align}
 &H_1.s_0=\overline{H}_1,\\
 &\overline{H}_2.s_0=H_1+H_2-E_1-E_4,
\end{align}
\end{subequations}
respectively indicate that 
the $(1,0)$-curve of $(f,g)$ is paired with the $(1,0)$-curve of $(\overline{f},\overline{g})$,
and the $(0,1)$-curve of $(\overline{f},\overline{g})$ is paired with 
the $(1,1)$-curve of $(f,g)$ passing through the base points $P_1$ and $P_4$,
we can set
\begin{subequations}\label{eqn:proof_1}
\begin{align}
 &\overline{f}(A_1f+A_2)+A_3f+A_4=0,\\
 &\overline{g}\left(B_1g(f+a_0)+B_2(f+{a_0}^{-1})\right)+B_3g(f+a_0)+B_4(f+{a_0}^{-1})=0.
\end{align}
\end{subequations}
Since the relations
\begin{equation}
 \overline{E}_i=E_i.s_0,\quad i=2,3,5,6,
\end{equation}
describe
\begin{equation}
 \left.\left(\overline{f},\overline{g}\right)\,\right|_{\overline{P}_i}
 =\left.\left(-\cfrac{A_3f+A_4}{A_1f+A_2}, -\cfrac{B_3g(f+a_0)+B_4(f+{a_0}^{-1})}{B_1g(f+a_0)+B_2(f+{a_0}^{-1})}\right)\,\right|_{P_i},\quad
 i=2,3,5,6,
\end{equation}
respectively, we obtain
\begin{equation}\label{eqn:proof_2}
 A_1=A_4=B_1=B_4=0,\quad
 B_2=-a_0 B_3,\quad
 \overline{a}_1=a_0a_1.
\end{equation}
Moreover, from
\begin{equation}
 \overline{E}_i=E_i.s_0,\quad
 i=7,8,
\end{equation}
we get
\begin{equation}
 \left.\left(\overline{f},\overline{g};\overline{f}\overline{g}\right)\,\right|_{\overline{P}_i}
 =\left.\left(-\cfrac{A_3}{A_2}\,f, g\,\cfrac{a_0+f}{1+a_0f}\,;-\cfrac{A_3}{A_2}\,\cfrac{a_0+f}{1+a_0f}\,fg\right)\,\right|_{P_i},\quad
 i=7,8,
\end{equation}
respectively, which lead to
\begin{equation}\label{eqn:proof_3}
 \cfrac{A_3}{A_2}=-\cfrac{a_0\overline{a}_0\overline{c}^2}{c^2},\quad
 \overline{a}_2=a_0a_2.
\end{equation}
Finally, from 
\begin{equation}
 \overline{H}_1-\overline{E}_4=E_1.s_0,\quad
 \overline{H}_1-\overline{E}_1=E_4.s_0,
\end{equation}
we obtain
\begin{equation}
 \overline{f}\,|_{\overline{P}_4}=\cfrac{a_0\overline{a}_0\overline{c}^2}{c^2}\,f\,|_{P_1},\quad
 \overline{f}\,|_{\overline{P}_1}=\cfrac{a_0\overline{a}_0\overline{c}^2}{c^2}\,f\,|_{P_4},
\end{equation}
respectively, which give
\begin{equation}\label{eqn:proof_4}
 \overline{a}_0={a_0}^{-1},\quad
 \overline{c}=c.
\end{equation}
Equations \eqref{eqn:proof_1}, \eqref{eqn:proof_2}, \eqref{eqn:proof_3} and \eqref{eqn:proof_4} provide the claim for $s_0$ \eqref{eqn:claim_s0}.
In a similar manner, we can prove the other birational actions.
\end{proof}

We can easily verify that the birational actions of $\widetilde{W}((A_2+A_1)^{(1)})$ 
given in Lemma \ref{lemma:birational_action_A2A1}
also satisfy the fundamental relations \eqref{eqn:fundamental_A2A1}.

\subsection{Discrete Painlev\'e equations}
It is well known that the translation part of $\widetilde{W}((A_2+A_1)^{(1)})$, which are translations in $Q((A_2+A_1)^{(1)})$,
give discrete Painlev\'e equations \cite{sak:01}.
We introduce the translations $T_i$, $i=1,2,3,4$, defined by
\begin{equation}
 T_1=\pi  s_2  s_1,\quad
 T_2=\pi  s_0  s_2,\quad
 T_3=\pi  s_1  s_0,\quad
 T_4=r  w_0.
\end{equation}
Note that $T_i$, $i =1,2,3,4$, commute with each other and $T_1T_2T_3=1$.
{\Rem
Observe that by convention \eqn{convd} we see that $T_1=t_1^{-1}=s_1s_2\rho^{-1}$, where $\pi=\rho^{-1}$,  
and $t_1$, defined in Equation \eqn{TaA2}, acts from the left
 is the translational element associated
with translation by the fundamental weight $h_1$ in $V^{(1)}$.
}

These are translation in the root system $Q((A_2+A_1)^{(1)})$:
\begin{subequations}
\begin{align}
 (\alpha_0,\alpha_1,\alpha_2,\beta_0,\beta_1).T_1
 =&(\alpha_2,\alpha_0,\alpha_1,\beta_0,\beta_1).s_2  s_1\notag\\
 =&(-\alpha_2,\alpha_0+\alpha_2,\alpha_1+\alpha_2,\beta_0,\beta_1).s_1\notag\\
 =&(\alpha_0,\alpha_1,\alpha_2,\beta_0,\beta_1)+(-1,1,0,0,0)\,\delta,\\
 (\alpha_0,\alpha_1,\alpha_2,\beta_0,\beta_1).T_2
 =&(\alpha_0,\alpha_1,\alpha_2,\beta_0,\beta_1)+(0,-1,1,0,0)\,\delta,\\
 (\alpha_0,\alpha_1,\alpha_2,\beta_0,\beta_1).T_3
 =&(\alpha_0,\alpha_1,\alpha_2,\beta_0,\beta_1)+(1,0,-1,0,0)\,\delta,\\
 (\alpha_0,\alpha_1,\alpha_2,\beta_0,\beta_1).T_4
 =&(\alpha_0,\alpha_1,\alpha_2,\beta_0,\beta_1)+(0,0,0,1,-1)\,\delta,
\end{align}
\end{subequations}
where we have used
\beqn
\delta=-K_X=\alpha_0+\alpha_1+\alpha_2=\beta_0+\beta_1,
\eeqn
and also these actions on the parameters are the translational motions as follows:
\begin{subequations}
\begin{align}
 T_1.(a_0,a_1,a_2,c)
 =&\pi  s_2.(a_0a_1,{a_1}^{-1},a_2a_1,c)
 =\pi.(qa_2,q^{-1}a_0,a_1,c)\notag\\
 =&(qa_0,q^{-1}a_1,a_2,c),\\
 T_2.(a_0,a_1,a_2,c)
 =&(a_0,qa_1,q^{-1}a_2,c),\\
 T_3.(a_0,a_1,a_2,c)
 =&(q^{-1}a_0,a_1,qa_2,c),\\
 T_4.(a_0,a_1,a_2,c)
 =&(a_0,a_1,a_2,qc).
\end{align}
\end{subequations}


Moreover, their actions on the variables are given by
\begin{subequations}
\begin{align}
 &T_i(f_i)=\cfrac{qc^2}{f_if_{i-1}}\,\cfrac{1+a_{i-1}f_{i-1}}{a_{i-1}+f_{i-1}},\quad
 T_i(f_{i-1})=\cfrac{qc^2}{f_{i-1}T_i(f_i)}\,\cfrac{1+a_{i-1}a_{i+1}T_i(f_i)}{a_{i-1}a_{i+1}+T_i(f_i)},\\
 &T_i(f_{i-2})=\cfrac{qc^2}{T_i(f_i)T_i(f_{i-1})},\quad
 T_4(f_i)=a_ia_{i+1} f_{i+1}\,\cfrac{1+a_{i-1} f_{i-1}(a_i f_i+1)}{1+a_i f_i(a_{i+1} f_{i+1}+1)},
\end{align}
\end{subequations}
where $i\in\mathbb{Z}/3\mathbb{Z}$.
The action of $T_1$ on $f_0$ and $f_1$ leads to a system of first-order ordinary difference equations:
\begin{equation}\label{eqn:qp3}
 G_{n+1}G_n=\cfrac{q c^2}{F_n}\,\cfrac{1+q^na_0F_n}{q^na_0+F_n},\quad
 F_{n+1}F_n=\cfrac{qc^2}{G_{n+1}}\,\cfrac{1+q^na_0a_2 G_{n+1}}{q^na_0a_2+G_{n+1}},
\end{equation}
where
\begin{equation}
 F_n={T_1}^n(f_0),\quad G_n={T_1}^n(f_1),
\end{equation}
which is known as a $q$-discrete analogue of Painlev\'e III equation
\cite{KTGR:00, sak:01}.
In a similar manner, in each of the $T_2$- and $T_3$-directions, we also obtain $q$-P$_{\rm III}$.
In contrast, the action of $T_4$ on $f_i$, $i=0,1,2$ and $f_0f_1f_2=qc^2$
lead to a $q$-discrete analogue of Painlev\'e IV equation \cite{KNY:01}:
\begin{subequations}\label{eqn:qp4}
\begin{align}
 &F_{n+1}= a_0a_1 G_n\cfrac{1+a_2 H_n( a_0 F_n+1)}{1+a_0 F_n(a_1 G_n+1)},\\
 &G_{n+1}=a_1a_2 H_n\cfrac{1+a_0 F_n(a_1 G_n+1)}{1+a_1 G_n(a_2 H_n+1)},\\
 &H_{n+1}=a_2a_0 F_n\cfrac{1+a_1 G_n(a_2 H_n+1)}{1+a_2 H_n(a_0 F_n+1)},\\
 &F_nG_nH_n=q^{2n+1}c^2,
\end{align}
\end{subequations}
where
\begin{equation}
 F_n={T_4}^n(f_0),\quad
 G_n={T_4}^n(f_1),\quad
 H_n={T_4}^n(f_2).
\end{equation}

It is also known that discrete dynamical systems of Painlev\'e type can be obtained 
from elements of infinite order of  $\widetilde{W}((A_2+A_1)^{(1)})$ which are not necessarily translations in $Q((A_2+A_1)^{(1)})$ \cite{Takenawa:03, KNT:11}.
We introduce the half-translation $R_1=\pi^2s_1$ satisfying 
\begin{equation}
 {R_1}^2=T_1.
\end{equation}
Transformation $R_1$ is not a translational motion on $Q((A_2+A_1)^{(1)})$ and parameter space:
\begin{subequations}
\begin{align}
 &(\alpha_0,\alpha_1,\alpha_2,\beta_0,\beta_1).R_1
 =(\alpha_0+\alpha_2-\delta,\alpha_1+\alpha_2,\delta-\alpha_2,\beta_0,\beta_1),\label{eqn:action_R1_rootA2A1}\\
 &R_1.(a_0,a_1,a_2,c)=(a_2a_0,q^{-1}a_2a_1,q{a_2}^{-1},c),
\end{align}
\end{subequations}
but by letting $a_2=q^{1/2}$ it becomes the translational motion in the parameter subspace:
\begin{equation}
 R_1.(a_0,a_1,c)=(q^{1/2}a_0,q^{-1/2}a_1,c).
\end{equation}
The action of $R_1$:
\begin{equation}
 R_1(f_1)=f_0,\quad
 R_1(f_0)=\cfrac{qc^2}{f_1f_0}\,\cfrac{1+a_0f_0}{a_0+f_0},
\end{equation}
with the condition $a_2=q^{1/2}$, gives the single second-order ordinary difference equation:
\begin{equation}\label{eqn:qp2}
 F_{n+1}F_{n-1}=\cfrac{qc^2}{F_n}\,\cfrac{1+q^{n/2}a_0F_n}{q^{n/2}a_0+F_n},
\end{equation}
where
\begin{equation}
 F_n={R_1}^n(f_0),
\end{equation}
which is known as a $q$-discrete analogue of Painlev\'e II equation \cite{RG:96}.

{\Rem
Let 
\begin{equation}
 \gamma_0=-\alpha_0+\alpha_1,\quad
 \gamma_1=2\alpha_0+\alpha_2,
\end{equation}
where $\delta=\gamma_0+\gamma_1$.
The submodule of $\bigoplus_{i=0}^2\mathbb{Z}\alpha_i$:
\begin{equation}
 \mathbb{Z}\gamma_0\bigoplus \mathbb{Z}\gamma_1
\end{equation}
is the root system of type $A_1^{(1)}$ since its corresponding Cartan matrix is of type $A_1^{(1)}$$:$
\begin{equation}
 (c_{ij})_{i,j=0}^1=\begin{pmatrix}2&-2\\-2&2\end{pmatrix},
\end{equation}
where
\begin{equation}
 c_{ij}=\cfrac{2(\gamma_i|\gamma_j)}{(\gamma_j|\gamma_j)}.
\end{equation}
It is obvious that the root system
\begin{equation}
 Q((A_1+A'_1)^{(1)})
 =\mathbb{Z}\gamma_0\bigoplus\mathbb{Z}\gamma_1\bigoplus\mathbb{Z}\beta_0\bigoplus\mathbb{Z}\beta_1
\end{equation}
is also orthogonal to the root system $Q(A_5^{(1)})$.
The transformation $R_1$ is not a translation in $Q((A_2+A_1)^{(1)})$ $($see Equation \eqref{eqn:action_R1_rootA2A1}$)$ 
but a translation in $Q((A_1+A'_1)^{(1)})$$:$
\begin{equation}
 (\gamma_0,\gamma_1,\beta_0,\beta_1).R_1
 =(\gamma_0,\gamma_1,\beta_0,\beta_1)+(1,-1,0,0)\delta.
\end{equation}
}
\section{Conclusion}\label{Con}
The purpose of this article is to give a detailed exposition about our approach
of associating discrete integrable systems to space-filling polytopes with the Weyl group symmetries.
The main outcome of this approach is that the connection between the
different discrete integrable systems can be clarified via the Weyl groups. These connections are
realised as reductions from higher dimensional system (quad-equations)
to lower dimensional system (discrete \Pa equations).
We associate the reductions with the degeneration of the symmetry
groups, which are realised as geometric constraints that give rise to deformation/degeneration 
of the polytopes of these symmetry groups, thus providing a simple way of obtaining and
understanding the connections between the different discrete integrable systems. 
Moreover, we have shown in Section 3 and 4
how the properties of the affine Weyl groups 
manifest in the two aspects of the discrete \Pa equations, giving complementary combinatorial and
geometrical information about the systems.

Having established the connections between quad-equations
and discrete \Pa equations some immediate consequences follow.
First, Lax pairs for discrete \Pa equations
can be derived using Lax pairs of the quad-equations \cite{jn3, jns5}.
Second, higher dimensional generalisations of discrete \Pa equations can be obtained by generalising the
combinatorics of the associated quad-equations and the reduction conditions. Such generalisation of system
\eqn{qPIII} was obtained in \cite{jns2}.

There are still many questions that remain open. For instance, all the 
discrete \Pa equations in Sakai's classification are of the simply laced types ($A$, $D$, $E$).
However, recently a nonlinear discrete system of $F_4$ type has been found as a reduction of the 
 Q4 equation in the ABS list \cite{ahjn:16}. There have not been many such examples, and their natures 
 are not yet well understood. 
 
 \noindent{\bf Acknowledgement.} 
Y. Shi would like to express her gratitude to R. B. Howlett 
for the enlightening discussions related to the Coxeter groups.

\bibliographystyle{abbrv}
\bibliography{wn_100815}

\end{document}